%% file: manu.tex
\pdfoutput=1
\documentclass[12pt,fleqn]{article}

\usepackage{amsmath,amssymb,amsthm,mathtools}
\usepackage{fullpage}
\usepackage{authblk}
\usepackage[parfill]{parskip}

\usepackage{comment}

\usepackage[T1]{fontenc}
\usepackage[scale=1.05]{biolinum}  
\usepackage{fouriernc}
\usepackage[cal=pxtx,scr=boondoxo,bb=stixtwo,frak=pxtx]{mathalpha}
\let\newmathbb\mathbb
\AtBeginDocument{%
    \let\mathbb\relax
    \newcommand{\mathbb}[1]{\newmathbb{#1}}
}
\renewcommand{\mathsf}[1]{\text{\upshape\sffamily#1}}
\newcommand{\upsigma}{\othersigma}
\newcommand{\uppsi}{\otherpsi}

\usepackage{cite}
\usepackage[colorlinks=true,citecolor=blue!65!black,linkcolor=red!75!black,backref=page]{hyperref}

\usepackage[nameinlink]{cleveref}

\input{macro}

\newtheorem{theorem}{Theorem}[section]
\newtheorem{proposition}[theorem]{Proposition}
\newtheorem{lemma}[theorem]{Lemma}
\newtheorem{claim}[theorem]{Claim}
\newtheorem{corollary}[theorem]{Corollary}
\newtheorem{observation}[theorem]{Observation}

\newtheorem{hypothesis}[theorem]{Hypothesis}
\theoremstyle{definition}
\newtheorem{definition}[theorem]{Definition}

\newtheorem{problem}[theorem]{Problem}
\newtheorem{example}[theorem]{Example}

\newenvironment{claim*}{\begin{claim}}{\end{claim}}
\newenvironment{claimproof}{\begin{proof}}{\end{proof}}

\numberwithin{equation}{section}

\newcommand{\acks}[1]{\paragraph{Acknowledgments.}#1}

\title{Gap Preserving Reductions Between \\ Reconfiguration Problems\footnote{
A preliminary version of this paper appeared in
\emph{Proc.~40th Int.~Symp.~on Theoretical Aspects of Computer Science (STACS)}, 2023 \cite{ohsaka2023gap}.
}}
\author{Naoto Ohsaka\thanks{CyberAgent, Inc., Tokyo, Japan. \href{mailto:ohsaka\_naoto@cyberagent.co.jp}{\texttt{ohsaka\_naoto@cyberagent.co.jp}}; \href{mailto:naoto.ohsaka@gmail.com}{\texttt{naoto.ohsaka@gmail.com}}
}}
\date{\today}

\begin{document}

\maketitle

\begin{abstract}\input{abstract}\end{abstract}

\tableofcontents

\clearpage


\input{main}

\acks
I wish to thank the anonymous referees for their suggestions which help improve the presentation of this paper.

\bibliographystyle{alpha}
\bibliography{manu}

\end{document}

%% file: macro.tex
\usepackage{booktabs,colortbl}
\usepackage{diagbox,enumitem}
\usepackage{framed,fancybox,ascmac}

\usepackage{comment}
\usepackage{url}

\usepackage[colorinlistoftodos,prependcaption]{todonotes}

\usepackage{xspace}
\usepackage{bm}

\usepackage{lineno}

\usepackage{tikz}
\usetikzlibrary{positioning}
\usetikzlibrary{decorations.markings}
\usetikzlibrary{shapes.geometric}
\usetikzlibrary{arrows.meta}
\usetikzlibrary{arrows,automata,decorations.pathmorphing,fit,positioning,arrows.meta,calc}

\crefname{theorem}{Theorem}{Theorems}
\crefname{proposition}{Proposition}{Propositions}
\crefname{lemma}{Lemma}{Lemmas}
\crefname{claim}{Claim}{Claims}
\crefname{corollary}{Corollary}{Corollaries}
\crefname{remark}{Remark}{Remarks}
\crefname{observation}{Observation}{Observations}
\crefname{hypothesis}{Hypothesis}{Hypotheses}
\crefname{definition}{Definition}{Definitions}
\crefname{problem}{Problem}{Problems}
\crefname{example}{Example}{Examples}

\crefname{appendix}{Appendix}{Appendices}
\crefname{section}{Section}{Sections}
\crefname{equation}{Eq.}{Eqs.}
\crefname{figure}{Figure}{Figures}
\crefname{table}{Table}{Tables}

\usepackage[ruled]{algorithm}
\usepackage{algpseudocodex}
\algnewcommand{\algorithmicand}{\textbf{ and }}
\algnewcommand{\algorithmicor}{\textbf{ or }}
\algnewcommand{\algorithmicto}{\textbf{ to }}

\algrenewcommand\textproc{\textsl}

\renewcommand{\geq}{\geqslant}
\renewcommand{\leq}{\leqslant}
\renewcommand{\phi}{\varphi}
\renewcommand{\epsilon}{\varepsilon}
\renewcommand{\bar}{\overline}

\newcommand{\prb}[1]{\textup{\textsc{#1}}\xspace}
\renewcommand{\prb}[1]{\textup{\textsf{#1}}\xspace}
\newcommand{\nth}[1]{#1\textsuperscript{th}\xspace}

\renewcommand{\vec}[1]{\mathbf{\bm{#1}}}

\newcommand{\reco}{\leftrightsquigarrow}

\DeclareMathOperator{\val}{\mathsf{val}}
\DeclareMathOperator{\cloud}{\mathsf{cloud}}
\DeclareMathOperator{\enc}{\mathsf{enc}}
\DeclareMathOperator{\PLR}{\mathsf{PLR}}
\DeclareMathOperator{\bigO}{\mathrm{O}}
\DeclareMathOperator*{\argmax}{argmax}

\newcommand{\sss}{\mathsf{s}}
\newcommand{\ttt}{\mathsf{t}}
\newcommand{\TTT}{\mathtt{T}}
\newcommand{\FFF}{\mathtt{F}}
\newcommand{\aaa}{\mathtt{a}}
\newcommand{\bbb}{\mathtt{b}}
\newcommand{\ccc}{\mathtt{c}}
\newcommand{\sigmaS}{\vec{\upsigma}}
\newcommand{\psiS}{\vec{\uppsi}}

\newcommand{\YES}{\textup{\textsc{yes}}\xspace}
\newcommand{\NO}{\textup{\textsc{no}}\xspace}
\newcommand{\scAND}{\textup{\textsc{and}}\xspace}
\newcommand{\scOR}{\textup{\textsc{or}}\xspace}
\newcommand{\scCHOICE}{\textup{\textsc{choice}}\xspace}
\newcommand{\scFANOUT}{\textup{\textsc{fanout}}\xspace}
\newcommand{\scREDBLUE}{\textup{\textsc{red--blue}}\xspace}

\newcommand{\cP}{\textup{\textbf{P}}\xspace}
\newcommand{\NP}{\textup{\textbf{NP}}\xspace}
\newcommand{\PSPACE}{\textup{\textbf{PSPACE}}\xspace}

\newcommand{\bbN}{\mathbb{N}}

\newcommand{\scrC}{\mathscr{C}}

\newcommand{\scrI}{\mathscr{I}}

\newcommand{\scrO}{\mathscr{O}}

\newcommand{\scrS}{\mathscr{S}}

\newcommand{\convexpath}[2]{
[   
    create hullnodes/.code={
        \global\edef\namelist{#1}
        \foreach [count=\counter] \nodename in \namelist {
            \global\edef\numberofnodes{\counter}
            \node at (\nodename) [draw=none,name=hullnode\counter] {};
        }
        \node at (hullnode\numberofnodes) [name=hullnode0,draw=none] {};
        \pgfmathtruncatemacro\lastnumber{\numberofnodes+1}
        \node at (hullnode1) [name=hullnode\lastnumber,draw=none] {};
    },
    create hullnodes
]
($(hullnode1)!#2!-90:(hullnode0)$)
\foreach [
    evaluate=\currentnode as \previousnode using \currentnode-1,
    evaluate=\currentnode as \nextnode using \currentnode+1
    ] \currentnode in {1,...,\numberofnodes} {
-- ($(hullnode\currentnode)!#2!-90:(hullnode\previousnode)$)
  let \p1 = ($(hullnode\currentnode)!#2!-90:(hullnode\previousnode) - (hullnode\currentnode)$),
    \n1 = {atan2(\y1,\x1)},
    \p2 = ($(hullnode\currentnode)!#2!90:(hullnode\nextnode) - (hullnode\currentnode)$),
    \n2 = {atan2(\y2,\x2)},
    \n{delta} = {-Mod(\n1-\n2,360)}
  in 
    {arc [start angle=\n1, delta angle=\n{delta}, radius=#2]}
}
-- cycle
}

%% file: abstract.tex
\emph{Combinatorial reconfiguration} is a growing research field studying
reachability and connectivity over the solution space of a combinatorial problem.
For example, in \prb{SAT Reconfiguration},
for a Boolean formula $\phi$ and its two satisfying truth assignments $\sigma_\sss$ and $\sigma_\ttt$,
we are asked to decide if
$\sigma_\sss$ can be transformed into $\sigma_\ttt$ by repeatedly
flipping a single variable assignment at a time,
while preserving every intermediate assignment satisfying $\phi$.
We consider the approximability of \emph{optimization variants} of reconfiguration problems; e.g.,
\prb{Maxmin SAT Reconfiguration} requires to maximize the minimum fraction of satisfied clauses of $\phi$ during transformation from $\sigma_\sss$ to $\sigma_\ttt$.
Solving such optimization variants approximately,
we may be able to acquire a reasonable transformation comprising
almost-satisfying truth assignments.

In this study,
we prove a series of \emph{gap-preserving reductions} to give evidence that
a host of reconfiguration problems are \PSPACE-hard to approximate,
under some plausible assumption.
Our starting point is a new working hypothesis called the \emph{Reconfiguration Inapproximability Hypothesis} (RIH),
which asserts that a gap version of \prb{Maxmin CSP Reconfiguration} is \PSPACE-hard.
This hypothesis may be thought of as a reconfiguration analogue of the PCP theorem \cite{arora1998probabilistic,arora1998proof}.
Our main result is 
\PSPACE-hardness of approximating \prb{Maxmin $3$-SAT Reconfiguration} of \emph{bounded occurrence} under RIH.
The crux of its proof is
a gap-preserving reduction from \prb{Maxmin Binary CSP Reconfiguration} to itself of \emph{bounded degree}.
Because a simple application of the degree reduction technique using expander graphs due to
{Papadimitriou and Yannakakis} (J.~Comput.~Syst.~Sci., 1991) \cite{papadimitriou1991optimization}
loses the \emph{perfect completeness},
we develop a new trick referred to as
\emph{alphabet squaring}, which modifies the alphabet as if each vertex could take a pair of values simultaneously.
To accomplish the soundness requirement, we further apply the expander mixing lemma and
an explicit family of near-Ramanujan graphs.
As an application of the main result, we demonstrate that under RIH,
optimization variants of popular reconfiguration problems are \PSPACE-hard to approximate,
including
\prb{Nondeterministic Constraint Logic} due to {Hearn and Demaine} (Theor.~Comput.~Sci., 2005) \cite{hearn2005pspace,hearn2009games},
\prb{Independent Set Reconfiguration},
\prb{Clique Reconfiguration},
\prb{Vertex Cover Reconfiguration}, and
\prb{$2$-SAT Reconfiguration}.
We finally highlight that all inapproximability results hold unconditionally as long as 
``\PSPACE-hard'' is replaced by ``\NP-hard.''

%% file: main.tex
\section{Introduction}
\emph{Combinatorial reconfiguration} is a growing research field studying 
reachability and connectivity over the solution space:
Given a pair of feasible solutions of a particular combinatorial problem, find
a step-by-step transformation from one to the other, called a \emph{reconfiguration sequence}.
Since the establishment of the unified framework of reconfiguration due to
{Ito, Demaine, Harvey, Papadimitriou, Sideri, Uehara, and Uno}~\cite{ito2011complexity},
numerous reconfiguration problems have been derived from source problems.
For example, in the canonical \prb{SAT Reconfiguration} problem \cite{gopalan2009connectivity},
we are given a Boolean formula $\phi$ and its two satisfying truth assignments $\sigma_\sss$ and $\sigma_\ttt$.
Then, we seek a reconfiguration sequence from $\sigma_\sss$ to $\sigma_\ttt$ composed only of satisfying truth assignments for $\phi$,
each resulting from the previous one by flipping a single variable assignment.\footnote{
Such a sequence forms a path on the Boolean hypercube.
}
Of particular importance is to reveal their computational complexity.
Most reconfiguration problems are classified as either
\cP (e.g.,
\prb{$3$-Coloring Reconfiguration} \cite{cereceda2011finding} and
\prb{Matching Reconfiguration} \cite{ito2011complexity}),
\NP-complete (e.g.,
\prb{Independent Set Reconfiguration} on bipartite graphs \cite{lokshtanov2019complexity}),
or
\PSPACE-complete (e.g.,
\prb{$3$-SAT Reconfiguration} \cite{gopalan2009connectivity} and
\prb{Independent Set Reconfiguration} \cite{hearn2005pspace}), and
recent studies dig into the fine-grained analysis using restricted graph classes and
parameterized complexity \cite{flum2006parameterized,downey2012parameterized}.
We refer the readers to surveys by
{van den Heuvel}~\cite{heuvel13complexity} and
{Nishimura}~\cite{nishimura2018introduction}
for more details.
One promising aspect has, however, been still less explored: \emph{approximability}.

Just like an \NP optimization problem derived from an \NP decision problem
(e.g., \prb{Max SAT} is a generalization of \prb{SAT}),
an \emph{optimization variant} can be defined for a reconfiguration problem,
which affords to \emph{relax} the feasibility of intermediate solutions.
For instance,
in \prb{Maxmin SAT Reconfiguration} \cite{ito2011complexity} --- an optimization variant of \prb{SAT Reconfiguration} --- we
wish to maximize the minimum fraction of clauses of $\phi$ satisfied
by any truth assignment during reconfiguration from $\sigma_\sss$ to $\sigma_\ttt$.
Such optimization variants naturally arise
when we are faced with the nonexistence of a reconfiguration sequence for the decision version, or
when we already know a problem of interest to be \PSPACE-complete.
Solving them approximately,
we may be able to acquire a reasonable reconfiguration sequence, e.g.,
that comprising \emph{almost-satisfying} truth assignments,
each violating at most $1\%$ of the clauses.

Indeed, in their seminal work,
{Ito et al.}~\cite{ito2011complexity}
proved inapproximability results of \prb{Maxmin SAT Reconfiguration} and \prb{Maxmin Clique Reconfiguration}, and
posed \PSPACE-hardness of approximation as an open problem.
Their results rely on \NP-hardness of the corresponding optimization problem,
which, however, does not bring us \PSPACE-hardness.
The significance of showing \PSPACE-hardness is that
it not only refutes a polynomial-time algorithm under \cP~$\neq$~\PSPACE,
but further disproves the existence of a witness (especially a reconfiguration sequence) of \emph{polynomial length}
under \NP~$\neq$~\PSPACE.
The present study aims to reboot the study on \PSPACE-hardness of approximation for reconfiguration problems,
assuming some plausible hypothesis.

\subsection{Our Working Hypothesis}
Since no \PSPACE-hardness of approximation for natural reconfiguration problems are known (to the best of our knowledge),
we assert a new working hypothesis called the \emph{Reconfiguration Inapproximability Hypothesis} (RIH),
concerning a gap version of \prb{Maxmin $q$-CSP Reconfiguration}, and
use it as a starting point.
\begin{hypothesis}[informal; see \cref{hyp:RIH}]
\begin{leftbar}
Given a constraint graph $G$ and its two satisfying assignments $\psi_\sss$ and $\psi_\ttt$,
it is \PSPACE-hard to distinguish between
\begin{itemize}
\item \YES instances, in which
    $\psi_\sss$ can be transformed into $\psi_\ttt$ by repeatedly changing the value of a single vertex at a time,
    while ensuring every intermediate assignment satisfying $G$, and
\item \NO instances, in which
    any such transformation induces an assignment violating $\epsilon$-fraction of the constraints.
\end{itemize}
\end{leftbar}
\end{hypothesis}
This hypothesis may be thought of as a reconfiguration analogue of the PCP theorem \cite{arora1998probabilistic,arora1998proof}, and
it already holds as long as ``\PSPACE-hard''  is replaced by ``\NP-hard'' \cite{ito2011complexity}.
Moreover, if a gap version of some optimization variant,
e.g., \prb{Maxmin SAT Reconfiguration},
is \PSPACE-hard,
RIH directly follows.
Our contribution is to demonstrate that the converse is also true:
Starting from RIH, we prove a series of (polynomial-time) \emph{gap-preserving reductions} to give evidence that
a host of reconfiguration problems are \PSPACE-hard to approximate.

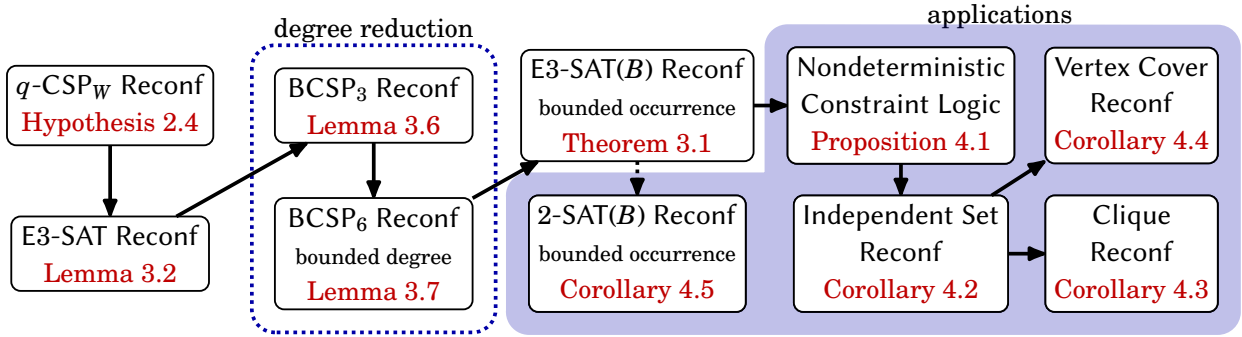
\begin{figure}[t]
    \centering
    \resizebox{\textwidth}{!}{\input{reductions}}%
    \caption{
        A series of gap-preserving reductions starting from the Reconfiguration Inapproximability Hypothesis used in this paper.
        Here,
        \prb{$q$-CSP$_W$ Reconf} and \prb{BCSP$_W$ Reconf} denote
        \prb{$q$-CSP Reconfiguration} and \prb{Binary CSP Reconfiguration}
        whose alphabet size is restricted to $W$, respectively;
        \prb{E$3$-SAT$(B)$ Reconf} denotes
        \prb{$3$-SAT Reconfiguration} in which
        every clause has exactly $3$ literals and
        each variable occurs in at most $B$ clauses.
        See \cref{sec:pre} for the formal definition of these problems.
        Note that all reductions
        excepting that for \prb{$2$-SAT$(B)$ Reconfiguration}
        (denoted dotted arrow)
        preserve the perfect completeness.
        Our results imply that
        approximating the above reconfiguration problems is \PSPACE-hard under RIH, and \NP-hard unconditionally.
    }
    \label{fig:reductions}
\end{figure}

\subsection{Our Results}
\cref{fig:reductions} presents an overall picture of the gap-preserving reductions introduced in this paper.
All reductions excepting \prb{$2$-SAT Reconfiguration}
preserve the \emph{perfect completeness}; i.e.,
\YES instances have a solution to the decision version.
Our main result is \PSPACE-hardness of approximating \prb{Maxmin E$3$-SAT Reconfiguration} of \emph{bounded occurrence} under RIH (\cref{thm:E3SAT}).
Here, ``bounded occurrence'' is critical to further reduce to \prb{Nondeterministic Constraint Logic},
which requires the number of clauses to be proportional to the number of variables.
Toward that end, we first reduce
\prb{Maxmin $q$-CSP Reconfiguration} to \prb{Maxmin Binary CSP Reconfiguration} in a gap-preserving manner
\emph{via} \prb{Maxmin E$3$-SAT Reconfiguration} (\cref{lem:qCSP-E3SAT,lem:E3SAT-BCSP}),
which employs a reconfigurable SAT encoding.

We then proceed to a gap-preserving reduction from \prb{Maxmin Binary CSP Reconfiguration} to itself of \emph{bounded degree} (\cref{lem:BCSP-deg_reduce}),
which is the most technical step in this paper.
Recall shortly the degree reduction technique due to 
{Papadimitriou and Yannakakis}~\cite{papadimitriou1991optimization},
also used by {Dinur}~\cite{dinur2007pcp} to prove the PCP theorem:
Each (high-degree) vertex is replaced by an expander graph called a \emph{cloud}, and 
equality constraints are imposed on the intra-cloud edges so that
the assignments in the cloud behave like a single assignment.
Observe easily that 
a simple application of this technique to \prb{Binary CSP Reconfiguration} loses the perfect completeness.
This is because we have to change the value of vertices in the cloud \emph{one by one}, sacrificing many equality constraints.
To bypass this issue,
we develop
a new trick referred to as \emph{alphabet squaring} tailored to reconfigurability,
which modifies the alphabet as if each vertex could take a pair of values simultaneously;
e.g., if the original alphabet is
$\Sigma = \{\aaa, \bbb, \ccc\}$,
the new one is
$\Sigma' = \{ \aaa, \bbb, \ccc, \aaa\bbb, \bbb\ccc, \ccc\aaa \}$.
Having a vertex to be assigned $\aaa\bbb$ represents that 
it has values $\aaa$ \emph{and} $\bbb$.
With this interpretation in mind, we redefine equality-like constraints for the intra-cloud edges so as to preserve the perfect completeness.

\begin{figure}[tbp]
    \centering
    \null\hfill
    \scalebox{0.8}{\input{example1}}
    \hfill
    \scalebox{0.8}{\input{example2}}
    \hfill\null
    \caption{
        A drawing of \cref{ex:BCSP}.
        The left side shows an instance $G$ of \prb{BCSP Reconfiguration}, where
        we cannot transform
            $\psi_\sss(w,v,x,y)=(\aaa, \aaa, \aaa, \aaa)$
            into
            $\psi_\ttt(w,v,x,y)=(\aaa, \aaa, \ccc, \ccc)$.
        The right side shows the resulting instance by applying
        the degree reduction step on $v$ of $G$.
        We can now assign conflicting values to $v_w$ and $v_x$ because edge $(v_w,v_x)$ does not exist;
        in particular, we can transform
            $\psi'_\sss(w,v_w,v_x,x,y)=(\aaa, \aaa, \aaa, \aaa, \aaa)$
            into
            $\psi'_\ttt(w,v_w,v_x,x,y)=(\aaa, \aaa, \aaa, \ccc, \ccc)$.
    }
    \label{fig:example}
\end{figure}
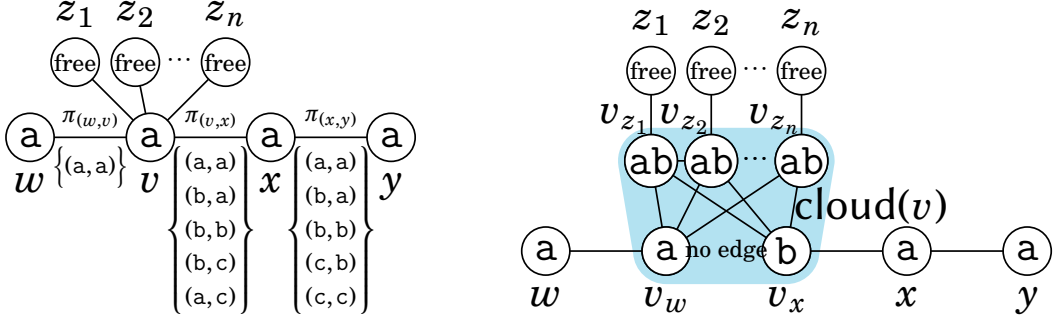

Unfortunately, using the alphabet squaring trick causes another issue, which renders the proof of the soundness requirement nontrivial.
\cref{ex:BCSP} illustrated in \cref{fig:example}
tells us that
our reduction is neither
a Karp reduction of \prb{Binary CSP Reconfiguration} nor
a PTAS reduction \cite{crescenzi2000approximation,crescenzi1997short} of \prb{Maxmin Binary CSP Reconfiguration}.
One particular reason is that assigning conflicting values to vertices in a cloud may not violate any equality-like constraints.
Thankfully,
we are ``promised'' that
at least $\epsilon$-fraction of constraints are unsatisfied during any transformation for some $\epsilon \in (0,1)$.
We thus use the following machinery to eventually accomplish the soundness requirement:

\begin{itemize}
\item
The crucial \cref{clm:BCSP-deg_reduce:ST} is that
for ``many'' vertices $v$,
there exists a pair of 
disjoint subsets $S_v$ and $T_v$ of $v$'s cloud such that their size is $\Theta(\epsilon \cdot |v\text{'s cloud}|)$ and
all constraints between them are unsatisfied.

\item Then, we apply
the \emph{expander mixing lemma} \cite{alon1988explicit} to bound the number of edges between $S_v$ and $T_v$ by
$\gtrapprox(d_0 \epsilon -  \lambda)\epsilon \cdot|v\text{'s cloud}|$,
where $d_0$ is the degree and
$\lambda$ is the second largest eigenvalue of 
$v$'s cloud.
Note that
{Papadimitriou and Yannakakis}~\cite{papadimitriou1991optimization} rely on the edge expansion property, 
which is not applicable as shown in \cref{ex:BCSP}.

\item We further use
an explicit family of \emph{near-Ramanujan graphs} \cite{alon2021explicit,mohanty2021explicit} so that the second largest eigenvalue $\lambda$ is $\bigO(\sqrt{d_0})$.
Setting the degree $d_0$ to $\bigO(\epsilon^{-2})$
ensures that $(d_0 \epsilon -  \lambda)\epsilon$
is positive constant; in particular,
the number of edges between $S_v$ and $T_v$ is
$\Theta(|v\text{'s cloud}|)$, as desired.
\end{itemize}
By applying this degree reduction step,
we come back to \prb{Maxmin E$3$-SAT Reconfiguration},
wherein, but this time, each variable appears in a constant number of clauses,
completing the proof of the main result.

Once we have established gap-preserving reducibility from RIH to \prb{Maxmin E$3$-SAT Reconfiguration} of bounded occurrence,
we can apply it to devise conditional \PSPACE-hardness of approximation for
an optimization variant of \prb{Nondeterministic Constraint Logic} (\cref{thm:NCL}).
\prb{Nondeterministic Constraint Logic} is a \PSPACE-complete problem 
proposed by {Hearn and Demaine} \cite{hearn2005pspace,hearn2009games}
that has been used to show \PSPACE-hardness of many games, puzzles, and 
other reconfiguration problems \cite{belmonte2021token,ito2012reconfiguration,bonsma2009finding,bousquet2022reconfiguration}.
We show that under RIH,
it is \PSPACE-hard to distinguish whether 
an input is a \YES instance, or
has a property that
every transformation must violate $\epsilon$-fraction of nodes.
The proof makes a modification to the existing gadgets \cite{hearn2005pspace,hearn2009games}.
As a consequence of \cref{thm:NCL},
we demonstrate that assuming RIH,
optimization variants of popular reconfiguration problems on graphs are \PSPACE-hard to approximate,
including
\prb{Independent Set Reconfiguration},
\prb{Clique Reconfiguration}, and
\prb{Vertex Cover Reconfiguration} (\cref{cor:NCL-IS,cor:IS-C,cor:IS-VC}),
whose proofs are almost immediate from existing work \cite{hearn2005pspace,hearn2009games,bonsma2009finding}.
We also show that \prb{Maxmin $2$-SAT Reconfiguration} of bounded occurrence
is \PSPACE-hard to approximate under RIH (\cref{cor:E3SAT-2SAT}),
whereas \prb{$2$-SAT Reconfiguration} belongs to \cP \cite{ito2011complexity}.
We finally highlight that all inapproximability results hold unconditionally as long as 
``\PSPACE-hard'' is replaced by ``\NP-hard.''

\subsection{Additional Related Work}
\label{subsec:intro:related}
Other reconfiguration problems whose approximability was analyzed include
\prb{Set Cover Reconfiguration}~\cite{ito2011complexity}, which is $2$-factor approximable,
\prb{Subset Sum Reconfiguration}~\cite{ito2014approximability}, which admits a PTAS,
\prb{Shortest Path Reconfiguration}~\cite{gajjar2022reconfiguring}, and
\prb{Submodular Reconfiguration}~\cite{ohsaka2022reconfiguration}.
The objective value of optimization variants is sometimes called
the \emph{reconfiguration index}~\cite{ito2016reconfiguration} or \emph{reconfiguration threshold}~\cite{de2018independent}.
We note that approximability of reconfiguration problems frequently refers to
that of \emph{the shortest sequence}
\cite{yamanaka2015swapping,miltzow2016approximation,bonnet2018complexity,heath2003sorting,bousquet2018reconfiguration,bousquet2020approximating,bousquet2019shortest,bonamy2020shortest,ito2022shortest}.
A different type of optimization variants,
called \emph{incremental optimization under the reconfiguration framework} \cite{ito2022incremental,blanche2020decremental,yanagisawa2021decremental}
has recently been studied; e.g.,
given an initial independent set, we want to transform it into a maximum possible independent set without touching those smaller than the specified size.
Those work seem orthogonal to the present study.

\section{Preliminaries}\label{sec:pre}

\paragraph{Notations.}
For a nonnegative integer $n \in \bbN$, let $ [n] \triangleq \{1, 2, \ldots, n\} $.
For a graph $G=(V,E)$, let $V(G)$ and $E(G)$ denote
the vertex set $V$ and edge set $E$ of $G$, respectively.
A \emph{sequence} $\scrS$ of a finite number of elements $S^{(0)}, S^{(1)}, \ldots, S^{(\ell)}$
is denoted by $\scrS = \langle S^{(0)}, S^{(1)}, \ldots, S^{(\ell)} \rangle$, and
we write $S^{(i)} \in \scrS$ to indicate that $S^{(i)}$ appears in $\scrS$.
We briefly recapitulate {Ito et al.}'s reconfiguration framework \cite{ito2011complexity}.
Suppose we are given
a ``definition'' of feasible solutions for some source problem and
a symmetric ``adjacency relation'' over a pair of feasible solutions.\footnote{
An adjacency relation can also be defined in terms of a ``reconfiguration step,''
which specifies how a solution can be transformed, e.g.,
a flip of a single variable assignment.
}
Then, for a pair of feasible solutions $S_\sss$ and $S_\ttt$,
a \emph{reconfiguration sequence from $S_\sss$ to $S_\ttt$} is
any sequence of feasible solutions,
$\scrS = \langle S^{(0)}, \ldots, S^{(\ell)} \rangle$,
starting from $S_\sss$ (i.e., $S^{(0)} = S_\sss$) and ending with $S_\ttt$ (i.e., $S^{(\ell)} = S_\ttt$) such that
all successive solutions $S^{(i-1)}$ and $S^{(i)}$ are adjacent.
In a reconfiguration problem,
we wish to decide if there exists a reconfiguration sequence between a pair of feasible solutions.

\subsection{Boolean Satisfiability and Reconfiguration}
We use the standard terminology and notation of Boolean satisfiability.
Truth values are denoted by $\TTT$ or $\FFF$.
A \emph{Boolean formula} $\phi$ consists of variables $x_1, \ldots, x_n$ and the logical operators,
AND ($\wedge$), OR ($\vee$), and NOT ($\neg$).
A \emph{truth assignment} $\sigma \colon \{x_1, \ldots, x_n\} \to \{\TTT, \FFF\}$ for $\phi$ is a mapping
that assigns a truth value to each variable.
A Boolean formula $\phi$ is said to be \emph{satisfiable} if 
there exists a truth assignment $\sigma$ such that
$\phi$ evaluates to $\TTT$ when each variable $x_i$ is assigned the truth value specified by $\sigma(x_i)$.
A \emph{literal} is either a variable or its negation; a \emph{clause} is a disjunction of literals.
A Boolean formula is in \emph{conjunctive normal form} (CNF) if
it is a conjunction of clauses.
A \emph{$k$-CNF} formula is a CNF formula in which every clause contains at most $k$ literals.
Hereafter, the prefix ``E$k$-'' means that every clause has exactly $k$ distinct literals, while
the suffix ``$(B)$'' indicates that
the number of occurrences of each variable is bounded by $B \in \bbN$.

Subsequently,
we formalize reconfiguration problems on Boolean satisfiability.
We say that two truth assignments for a Boolean formula are \emph{adjacent} if 
one is obtained from the other by flipping a single variable assignment; i.e.,
they differ in exactly one variable.
The \prb{$k$-SAT Reconfiguration} problem \cite{gopalan2009connectivity}
is a decision problem of determining
for a $k$-CNF formula $\phi$ and its two satisfying truth assignments $\sigma_\sss$ and $\sigma_\ttt$,
whether there is a reconfiguration sequence of satisfying truth assignments for $\phi$ from $\sigma_\sss$ to $\sigma_\ttt$.
Since we are concerned with approximability of reconfiguration problems,
we formulate its optimization variant \cite{ito2011complexity}, which
allows us to employ \emph{non-satisfying} truth assignments.
For a CNF formula $\phi$ consisting of $m$ clauses $C_1, \ldots, C_m$ and a truth assignment $\sigma$ for $\phi$,
let $\val_\phi(\sigma)$ denote the fraction of clauses of $\phi$ satisfied by $\sigma$; namely,
\begin{align}
    \val_\phi(\sigma) \triangleq
    \frac{\left|\{ j \in [m] \mid \sigma \text{ satisfies } C_j \}\right|}{m}.
\end{align}
For a reconfiguration sequence of truth assignments  for $\phi$,
$\sigmaS = \langle \sigma^{(0)}, \ldots, \sigma^{(\ell)} \rangle$,
let $\val_\phi(\sigmaS)$ denote the \emph{minimum} fraction of satisfied clauses  of $\phi$
over all $\sigma^{(i)}$'s in $\sigmaS$; namely,
\begin{align}
    \val_\phi(\sigmaS) \triangleq \min_{\sigma^{(i)} \in \sigmaS} \val_\phi(\sigma^{(i)}).
\end{align}
Then,
for a $k$-CNF formula $\phi$ and
its truth assignments $\sigma_\sss$ and $\sigma_\ttt$
(which are not necessarily satisfying),
\prb{Maxmin $k$-SAT Reconfiguration}
is defined as an optimization problem of
maximizing $\val_\phi(\sigmaS)$ subject to $\sigmaS = \langle \sigma_\sss, \ldots, \sigma_\ttt \rangle$.
Observe that \prb{Maxmin $k$-SAT Reconfiguration} is \PSPACE-hard
because so is \prb{$k$-SAT Reconfiguration}~\cite{gopalan2009connectivity}.
Let $\val_\phi(\sigma_\sss \reco \sigma_\ttt)$ denote the maximum value of $\val_\phi(\sigmaS)$
over all possible reconfiguration sequences $\sigmaS$ from $\sigma_\sss$ to $\sigma_\ttt$; namely,
\begin{align}
    \val_\phi(\sigma_\sss \reco \sigma_\ttt) \triangleq
    \max_{\sigmaS = \langle \sigma_\sss, \ldots, \sigma_\ttt \rangle } \val_\phi(\sigmaS)
    = \max_{\sigmaS=\langle \sigma_\sss, \ldots, \sigma_\ttt \rangle}
    \min_{\sigma^{(i)} \in \sigmaS} \val_\phi(\sigma^{(i)}).
\end{align}
Note that $\val_\phi(\sigma_\sss \reco \sigma_\ttt) \leq \min\{ \val_\phi(\sigma_\sss), \val_\phi(\sigma_\ttt) \}$.
If $\val_\phi(\sigma_\sss \reco \sigma_\ttt) \geq \rho$ for some $\rho$,
we can transform $\sigma_\sss$ into $\sigma_\ttt$ while ensuring
that \emph{every} intermediate truth assignment satisfies at least $\rho$-fraction of the clauses of $\phi$.
The \emph{gap version} of \prb{Maxmin $k$-SAT Reconfiguration}
is finally defined as follows:

\begin{problem}\label{prb:gap-SAT}
\begin{leftbar}
For every $k \in \bbN$ and $0 \leq s \leq c \leq 1$,
\prb{Gap$_{c,s}$ $k$-SAT Reconfiguration} requests to distinguish
for a $k$-CNF formula $\phi$ and two (not necessarily satisfying) truth assignments $\sigma_\sss$ and $\sigma_\ttt$ for $\phi$,
whether
$\val_\phi(\sigma_\sss \reco \sigma_\ttt) \geq c$ (the input is a \YES instance) or
$\val_\phi(\sigma_\sss \reco \sigma_\ttt) < s$ (the input is a \NO instance).
Here, $c$ and $s$ denote \emph{completeness} and \emph{soundness}, respectively.
\end{leftbar}
\end{problem}
\cref{prb:gap-SAT} is a \emph{promise problem}, in which
we can output anything when
$s \leq \val_\phi(\sigma_\sss \reco \sigma_\ttt) < c$.
The present definition does not request an actual reconfiguration sequence.
Note that we can assume $\sigma_\sss$ and $\sigma_\ttt$ to be satisfying ones whenever $c=1$, and
the case of $s = c = 1$ particularly reduces to \prb{$k$-SAT Reconfiguration}.

\subsection{Constraint Satisfaction Problem and Reconfiguration}
Let us first define the notion of \emph{constraint graphs}.

\begin{definition}
\begin{leftbar}
A $q$-ary \emph{constraint graph} is defined as a tuple $G=(V,E,\Sigma,\Pi)$, such that
\begin{itemize}
    \item $(V,E)$ is a $q$-uniform hypergraph called the \emph{underlying graph},
    \item $\Sigma$ is a finite set called the \emph{alphabet}, and
    \item $\Pi = (\pi_e)_{e \in E}$ is a collection of $q$-ary \emph{constraints}, where
    each constraint $\pi_e \subseteq \Sigma^e$ is a set of $q$-tuples of acceptable values that $q$ vertices in $e$ can take.
\end{itemize}
The \emph{degree} $d_G(v)$ of each vertex $v$ in $G$ is defined as the number of hyperedges including $v$.
\end{leftbar}
\end{definition}
For a $q$-ary constraint graph $G=(V,E,\Sigma,\Pi = (\pi_e)_{e \in E})$,
an \emph{assignment} is a mapping $\psi \colon V \to \Sigma$ that assigns a value of $\Sigma$ to each vertex of $V$.
We say that
$\psi$ \emph{satisfies} hyperedge $e = \{v_1, \ldots, v_q\} \in E$ (or constraint $\pi_e$) if
$\psi(e) \triangleq (\psi(v_1), \ldots, \psi(v_q)) \in \pi_e$,
$\psi$ \emph{satisfies} $G$ if it satisfies all hyperedges of $G$, and
$G$ is \emph{satisfiable} if there exists an assignment that satisfies $G$.
Recall that \prb{$q$-CSP} requires to decide if 
a $q$-ary constraint graph is satisfiable.
Hereafter, \prb{BCSP} stands for \prb{$2$-CSP},
\prb{$q$-CSP$_W$} designates the restricted case that the alphabet size $|\Sigma|$ is some $W \in \bbN$, and
\prb{$q$-CSP$(\Delta)$} for some $\Delta \in \bbN$ means that the maximum degree of the constraint graph is bounded by $\Delta$.

We then proceed to reconfiguration problems on constraint satisfaction.
Two assignments are \emph{adjacent} if they differ in exactly one vertex.
In \prb{$q$-CSP Reconfiguration},
for a $q$-ary constraint graph $G$ and its two satisfying assignments $\psi_\sss$ and $\psi_\ttt$,
we are asked to decide if there is a reconfiguration sequence of satisfying assignments for $G$ from $\psi_\sss$ to $\psi_\ttt$.
Then, analogously to the case of Boolean satisfiability,
we introduce the following notations:
\begin{align}
    \val_G(\psi) \triangleq
    \frac{\left|\{e \in E \mid \psi \text{ satisfies } e \}\right|}{|E|}
\end{align}
for assignment $\psi \colon V \to \Sigma$,
\begin{align}
    \val_G(\psiS) \triangleq \min_{\psi^{(i)} \in \psiS} \val_G(\psi^{(i)})    
\end{align}
for reconfiguration sequence $\psiS = \langle \psi^{(i)} \rangle_{0 \leq i \leq \ell}$, and
\begin{align}
    \val_G(\psi_\sss \reco \psi_\ttt) \triangleq \max_{\psiS = \langle \psi_\sss, \ldots, \psi_\ttt \rangle} \val_G(\psiS)  
\end{align}
for two assignments $\psi_\sss, \psi_\ttt \colon V \to \Sigma$.
For a pair of assignments $\psi_\sss$ and $\psi_\ttt$ for $G$,
\prb{Maxmin $q$-CSP Reconfiguration} requests to maximize
$\val_G(\psiS)$ subject to $\psiS = \langle \psi_\sss, \ldots, \psi_\ttt \rangle$, while
its gap version is defined below.
\begin{problem}
\begin{leftbar}
For every $q \in \bbN$ and $0 \leq s \leq c \leq 1$,
\prb{Gap$_{c,s}$ $q$-CSP Reconfiguration} requests to distinguish
for a $q$-ary constraint graph $G$ and two (not necessarily satisfying) assignments $\psi_\sss$ and $\psi_\ttt$ for $G$,
whether
$\val_G(\psi_\sss \reco \psi_\ttt) \geq c$ or
$\val_G(\psi_\sss \reco \psi_\ttt) < s$.
\end{leftbar}
\end{problem}

\paragraph{Reconfiguration Inapproximability Hypothesis.}
We now present a formal description of our working hypothesis,
which serves as a starting point for \PSPACE-hardness of approximation.
\begin{hypothesis}[Reconfiguration Inapproximability Hypothesis, RIH]
\label{hyp:RIH}
\begin{leftbar}
There exist universal constants
$q, W \in \bbN$ and
$\epsilon \in (0,1)$ such that
\prb{Gap$_{1, 1-\epsilon}$ $q$-CSP$_W$ Reconfiguration} is \PSPACE-hard.
\end{leftbar}
\end{hypothesis}
Note that \NP-hardness of 
\prb{Gap$_{1, 1-\epsilon}$ $q$-CSP$_W$ Reconfiguration}
was already shown \cite{ito2011complexity}.

\section{Hardness of Approximation for {\normalfont \prb{Maxmin E$3$-SAT$(B)$ Reconfiguration}}}
\label{sec:SAT-CSP}
In this section, we prove the main result of this paper; that is,
\prb{Maxmin E$3$-SAT Reconfiguration} of bounded occurrence is \PSPACE-hard to approximate under RIH.

\begin{theorem}\label{thm:E3SAT}
\begin{leftbar}
Under \cref{hyp:RIH},
there exist universal constants $B \in \bbN$ and $\epsilon \in (0,1)$ such that 
\prb{Gap$_{1,1-\epsilon}$ E$3$-SAT$(B)$ Reconfiguration} is \PSPACE-hard.
\end{leftbar}
\end{theorem}
The remainder of this section is devoted to the proof of \cref{thm:E3SAT} and organized as follows:
In \cref{subsec:SAT-CSP:qCSP-E3SAT-BCSP}, we reduce
\prb{Maxmin $q$-CSP$_W$ Reconfiguration} to \prb{Maxmin BCSP$_3$ Reconfiguration},
\cref{subsec:SAT-CSP:BCSP-degreduce} presents the degree reduction of \prb{Maxmin BCSP Reconfiguration}, and
\cref{subsec:SAT-CSP:together} concludes the proof of \cref{thm:E3SAT}.

\subsection{Gap-preserving Reduction from {\normalfont \prb{Maxmin $q$-CSP$_W$ Reconfiguration}} to {\normalfont \prb{Maxmin BCSP$_3$ Reconfiguration}}}
\label{subsec:SAT-CSP:qCSP-E3SAT-BCSP}

We first reduce \prb{Maxmin $q$-CSP$_W$ Reconfiguration} to \prb{Maxmin E$3$-SAT Reconfiguration}.

\begin{lemma}\label{lem:qCSP-E3SAT}
\begin{leftbar}
For every $q, W \geq 2$ and $\epsilon \in (0,1)$,
there exists a gap-preserving reduction from
\prb{Gap$_{1,1-\epsilon}$ $q$-CSP$_W$ Reconfiguration} to
\prb{Gap$_{1,1-\epsilon'}$ E$3$-SAT Reconfiguration}, where
$\epsilon' = \frac{\epsilon}{W^q \cdot 2^{qW} (qW-2)}$.
Moreover,
if the maximum degree of the constraint graph in the former problem is $\Delta$,
then the number of occurrences of each variable in the latter problem is bounded by $W^q \cdot 2^{qW} \Delta$.
\end{leftbar}
\end{lemma}
The proof of \cref{lem:qCSP-E3SAT} consists of
a reduction from 
\prb{Maxmin $q$-CSP$_W$ Reconfiguration} to \prb{Maxmin E$k$-SAT Reconfiguration},
where the clause size $k$ depends solely on $q$ and $W$, and
that from \prb{Maxmin E$k$-SAT Reconfiguration} to \prb{Maxmin E$3$-SAT Reconfiguration}.

\begin{claim}\label{lem:qCSP-EkSAT}
\begin{leftbar}
For every $q, W \geq 2$ and $\epsilon \in (0,1)$,
there exists a gap-preserving reduction from
\prb{Gap$_{1, 1-\epsilon}$ $q$-CSP$_W$ Reconfiguration} to
\prb{Gap$_{1,1-\frac{\epsilon}{W^q \cdot 2^{qW}}}$ E$k$-SAT Reconfiguration},
where $k = qW$.
Moreover, if the maximum degree of the constraint graph in the former problem is $\Delta$, then
the number of occurrences of each variable in the latter problem is bounded by $W^q \cdot 2^{qW} \Delta$.
\end{leftbar}
\end{claim}

\begin{claim}\label{lem:EkSAT-E3SAT}
\begin{leftbar}
For every $k \geq 4$ and $\epsilon \in (0,1)$,
there exists a gap-preserving reduction from
\prb{Gap$_{1,1-\epsilon}$ E$k$-SAT Reconfiguration} to
\prb{Gap$_{1,1-\frac{\epsilon}{k-2}}$ E$3$-SAT Reconfiguration}.
Moreover, if the number of occurrences of each variable in the former problem is $B$,
then the number of occurrences of each variable in the latter problem is bounded by $\max\{B, 2\}$.
\end{leftbar}
\end{claim}
\cref{lem:qCSP-E3SAT} follows from \cref{lem:qCSP-EkSAT,lem:EkSAT-E3SAT}.

\begin{table}[tbp]
    \centering
    \begin{tabular}{c|c}
        \toprule
        $\vec{s} \in \{\TTT,\FFF\}^\Sigma$ & $\enc(\vec{s}) \in \Sigma$ \\
        \midrule
        $\FFF\FFF\FFF$ & $1$ \\
        $\TTT\FFF\FFF$ & $1$ \\
        $\FFF\TTT\FFF$ & $2$ \\
        $\TTT\TTT\FFF$ & $2$ \\
        $\FFF\FFF\TTT$ & $3$ \\
        $\TTT\FFF\TTT$ & $3$ \\
        $\FFF\TTT\TTT$ & $3$ \\
        $\TTT\TTT\TTT$ & $3$ \\
        \bottomrule
    \end{tabular}
    \caption{Example of $\enc \colon \{\TTT, \FFF\}^\Sigma \to \Sigma$ when $\Sigma = [3]$.}
    \label{tab:encoding}
\end{table}

\paragraph{Reconfigurable SAT Encoding.}
For the proof of \cref{lem:qCSP-EkSAT},
we introduce a slightly sophisticated SAT encoding of the alphabet.
Hereafter, we denote $\Sigma \triangleq [W]$ for some $W \in \bbN$.
Consider an encoding $\enc \colon \{\TTT,\FFF\}^\Sigma \to \Sigma$
of a binary string $\vec{s} \in \{\TTT, \FFF\}^\Sigma$ to $\Sigma$ defined as follows:
\begin{align}
\textsf{enc}(\vec{s}) \triangleq
\begin{cases}
    1 & \mbox{if } s_\alpha = \FFF \mbox{ for all } \alpha \in \Sigma, \\
    \alpha & \mbox{if } s_\alpha = \TTT \mbox{ and }  s_{\beta} = \FFF \mbox{ for all } \beta > \alpha.
\end{cases}
\end{align}
See \cref{tab:encoding} for an example of $\enc$ for $\Sigma=[3]$.
$\enc$ exhibits the following property concerning reconfigurability:

\begin{claim}\label{clm:string-reconf}
\begin{leftbar}
For any two strings $\vec{s}$ and $\vec{t}$ in $\{\TTT, \FFF\}^\Sigma$
with $\alpha \triangleq \enc(\vec{s})$ and $\beta \triangleq \enc(\vec{t})$,
we can transform $\vec{s}$ into $\vec{t}$
by repeatedly flipping one entry at a time while preserving every intermediate string mapped to $\alpha$ or $\beta$ by $\enc$.
\end{leftbar}
\end{claim}
\begin{claimproof}
The proof is done by induction on the size $W$ of $\Sigma$.
The case of $W = 1$ is trivial.
Suppose the statement holds for $W-1$.
Let $\vec{s}$ and $\vec{t}$ be any two strings such that $\alpha = \enc(\vec{s})$ and $\beta = \enc(\vec{t})$.
The case of $\alpha, \beta < W$ reduces to the induction hypothesis.
If $\alpha = \beta = W$, then $\vec{s}$ and $\vec{t}$ are reconfigurable to each other because
any string $\vec{u} \in \{\TTT, \FFF\}^\Sigma$ satisfies $\enc(\vec{u}) = W$ \emph{if and only if} $u_W = \TTT$.
Consider now the case that $\alpha = W$ and $\beta < W$ without loss of generality.
We can easily transform $\vec{s}$ into the string $\vec{s}' \in \{\TTT, \FFF\}^\Sigma$ such that
\begin{align}
    s'_\gamma =
    \begin{cases}
    \TTT & \text{if } \gamma = W, \\
    t_\gamma & \text{if } \gamma \leq W-1.
    \end{cases}
\end{align}
Observe that $\vec{s}'$ and $\vec{t}$ differ in only one entry,
which completes the proof.
\end{claimproof}
In the proof of \cref{lem:qCSP-EkSAT},
we use $\enc$ to encode each $q$-tuple of unacceptable values $(\alpha_1, \ldots, \alpha_q) \allowbreak \in \Sigma^e \setminus \pi_e$
for hyperedge $e = \{v_1, \ldots, v_q\} \in E$.

\begin{claimproof}[Proof of \cref{lem:qCSP-EkSAT}]
We first describe a gap-preserving reduction from
\prb{Maxmin $q$-CSP$_W$ Reconfiguration} to \prb{Maxmin E$k$-SAT Reconfiguration}.
Let $(G, \psi_\sss, \psi_\ttt)$ be an instance of \prb{Maxmin $q$-CSP$_W$ Reconfiguration},
where $G=(V,E,\Sigma=[W],\Pi = (\pi_e)_{e \in E})$ is a $q$-ary constraint graph, and
$\psi_\sss$ and $\psi_\ttt$ satisfy $G$.
For each vertex $v \in V$ and value $\alpha \in \Sigma$,
we create a variable $x_{v,\alpha}$.
Let $V'$ denote the set of the variables; i.e.,
$V' \triangleq \{x_{v,\alpha} \mid v \in V, \alpha \in \Sigma\}$.
Thinking of $(x_{v,1}, x_{v,2}, \ldots, x_{v,W})$ as a vector of $W$ variables,
we denote $\vec{x}_v \triangleq (x_{v,\alpha})_{\alpha \in \Sigma}$.
By abuse of notation,
we write $\sigma(\vec{x}_v) \triangleq (\sigma(x_{v,1}), \sigma(x_{v,2}), \ldots, \sigma(x_{v,W}))$
for truth assignment $\sigma \colon V' \to \{\TTT, \FFF\}$.
Then,
for each hyperedge $e = \{v_1, \ldots, v_q\} \in E$,
we will construct a CNF formula $\phi_e$ that emulates constraint $\pi_e$.
In particular, 
for each $q$-tuple of \emph{unacceptable} values $(\alpha_1, \ldots, \alpha_q) \in \Sigma^e \setminus \pi_e$,
$\phi_e$ should prevent 
$(\enc(\sigma(\vec{x}_{v_1})), \ldots, \enc(\sigma(\vec{x}_{v_q})))$
from being equal to
$(\alpha_1, \ldots, \alpha_q)$ for $\sigma \colon V' \to \{\TTT,\FFF\}$; that is,
we shall ensure
\begin{align}
\label{eq:CSP-SAT-nae}
    \bigvee_{i \in [q]} \Bigl( \enc(\sigma(\vec{x}_{v_i})) \neq \alpha_i \Bigr).
\end{align}
Such a CNF formula can be obtained by the following procedure:
\begin{itembox}[l]{\textbf{Construction of a CNF formula $\phi_e$}}
\begin{algorithmic}[1]
    \State initialize an empty CNF formula $\phi_e$.
    \For{\textbf{each} $q$-tuple of unacceptable values $(\alpha_1, \ldots, \alpha_q) \in \Sigma^e \setminus \pi_e$}
        \For{\textbf{each} $q$-tuple of vectors $\vec{s}_1, \ldots, \vec{s}_q \in \{\TTT, \FFF\}^\Sigma$ s.t.~$\enc(\vec{s}_i) = \alpha_i$ for all $i \in [q]$}
        \State add the following clause to $\phi_e$:
    \begin{align}
    \bigvee_{\alpha \in \Sigma} \bigvee_{i \in [q]} \llbracket x_{v_i, \alpha} \neq s_{i, \alpha} \rrbracket,
    \text{  where  }
    \llbracket x_{v_i, \alpha} \neq s_{i, \alpha} \rrbracket \triangleq
    \begin{cases}
        x_{v_i, \alpha} & \text{ if } s_{i, \alpha} = \FFF, \\
        \bar{x_{v_i, \alpha}} & \text{ if } s_{i, \alpha} = \TTT.
    \end{cases}
    \end{align}
        \EndFor
    \EndFor
    \State \textbf{return} $\phi_e$.
\end{algorithmic}
\end{itembox}
The resulting CNF formula $\phi_e$ thus looks like
\begin{align}
    \bigwedge_{(\alpha_1, \ldots, \alpha_q) \in \Sigma^e \setminus \pi_e}
    \bigwedge_{\substack{\vec{s}_1, \ldots, \vec{s}_q \in \{\TTT,\FFF\}^\Sigma: \\ \enc(\vec{s}_i) = \alpha_i \forall i \in [q]}}
    \bigvee_{\alpha \in \Sigma} \bigvee_{i \in [q]}
    \llbracket x_{v_i,\alpha} \neq s_{i,\alpha} \rrbracket.
\end{align}
Observe that a truth assignment $\sigma \colon V' \to \{\TTT, \FFF\}$ makes all clauses of $\phi_e$ true if and only if
an assignment $\psi \colon V \to \Sigma$, such that $\psi(v) \triangleq \enc(\sigma(\vec{x}_v))$ for all $v \in V$, satisfies $\pi_e$.
Define $\phi \triangleq \bigwedge_{e \in E} \phi_e$ to complete the construction of $\phi$.
For a satisfying assignment $\psi \colon V \to \Sigma$ for $G$,
let $\sigma_\psi \colon V' \to \{\TTT, \FFF\}$ be a truth assignment for $\phi$ such that
$\sigma_\psi(\vec{x}_v)$ for each vertex $v \in V$
is the lexicographically smallest string with
$\enc(\sigma_\psi(\vec{x}_v)) = \psi(v)$.
Then, $\sigma_\psi$ satisfies $\phi$.
Constructing $\sigma_\sss$ from $\psi_\sss$ and $\sigma_\ttt$ from $\psi_\ttt$
according to this procedure,
we obtain an instance $(\phi, \sigma_\sss, \sigma_\ttt)$ of \prb{Maxmin $k$-SAT Reconfiguration}, which completes the reduction.
Note that
the number of clauses $m$ in $\phi$ is
\begin{align}
m \leq \sum_{e \in E} |\Sigma^e \setminus \pi_e| \cdot 2^{|e|W} \leq W^q \cdot 2^{qW} |E|,
\end{align}
the size of every clause is exactly $k = qW$, and
each variable appears in at most $W^q \cdot 2^{qW} \Delta$ clauses of $\phi$ if the maximum degree of $G$ is $\Delta$.

We first prove the completeness; i.e.,
$\val_G(\psi_\sss \reco \psi_\ttt) = 1$ implies $\val_\phi(\sigma_\sss \reco \sigma_\ttt) = 1$,
using \cref{clm:string-reconf}.
It suffices to consider the case that $\psi_\sss$ and $\psi_\ttt$ differ in exactly one vertex, say, $v \in V$.
Since 
$\enc(\sigma_\sss(\vec{x}_v)) = \psi_\sss(v) \neq \psi_\ttt(v) = \enc(\sigma_\ttt(\vec{x}_v))$,
it holds that $\sigma_\sss(\vec{x}_v) \neq \sigma_\ttt(\vec{x}_v)$.
On the other hand, it holds that
$\sigma_\sss(\vec{x}_{w}) = \sigma_\ttt(\vec{x}_{w})$ for all $w \neq v$.
By \cref{clm:string-reconf},
we can find a sequence of strings in $\{\TTT,\FFF\}^\Sigma$,
$\langle \vec{s}^{(0)} = \sigma_\sss(\vec{x}_v), \ldots, \vec{s}^{(\ell)} = \sigma_\ttt(\vec{x}_v) \rangle$,
such that
two successive strings differ in exactly one entry, and
each intermediate $\enc(\vec{s}^{(i)})$ is equal to either $\enc(\sigma_\sss(\vec{x}_v))$ or $\enc(\sigma_\ttt(\vec{x}_v))$.
Using this string sequence, we construct another sequence of assignments,
$\sigmaS = \langle \sigma^{(i)} \rangle_{0 \leq i \leq \ell}$,
where each $\sigma^{(i)} \colon V' \to \{\TTT,\FFF\}$
is obtained from $\sigma_\sss$ by replacing assignments to $\vec{x}_v$ by $\vec{s}^{(i)}$;
namely,
$\sigma^{(i)}(\vec{x}_v) \triangleq \vec{s}^{(i)}$ whereas
$\sigma^{(i)}(\vec{x}_w) \triangleq \sigma_\sss(\vec{x}_w) = \sigma_\ttt(\vec{x}_w)$ for all $w \neq v$.
Observe easily that
$\sigmaS$ is a valid reconfiguration sequence for $(\phi, \sigma_\sss, \sigma_\ttt)$, and
each $\sigma^{(i)}$ satisfies $\phi$ because 
$\enc(\sigma^{(i)}(\vec{x}_w))$ is $\enc(\sigma_\sss(\vec{x}_w)) $ or $\enc(\sigma_\ttt(\vec{x}_w))$ for all $w \in V$;
i.e., $\val_\phi(\sigmaS) = 1$, as desired.

We then prove the soundness; i.e.,
$\val_G(\psi_\sss \reco \psi_\ttt) < 1-\epsilon$ implies
$\val_\phi(\sigma_\sss \reco \sigma_\ttt) < 1 - \frac{\epsilon}{W^q \cdot 2^{qW}}$.
Let
$\sigmaS = \langle \sigma^{(0)} = \sigma_\sss, \ldots, \sigma^{(\ell)} = \sigma_\ttt \rangle$
be any reconfiguration sequence for $(\phi, \sigma_\sss, \sigma_\ttt)$.
Construct then a sequence of assignments,
$\psiS = \langle \psi^{(i)} \rangle_{0 \leq i \leq \ell}$,
where each $\psi^{(i)} \colon V \to \Sigma$ is defined as
$\psi^{(i)}(v) \triangleq \enc(\sigma^{(i)}(\vec{x}_v))$ for all $v \in V$.
Since $\psiS$ is a valid reconfiguration sequence for $(G, \psi_\sss, \psi_\ttt)$,
we have $\val_G(\psiS) < 1-\epsilon$;
in particular,
there exists some $\psi^{(i)}$ such that $\val_G(\psi^{(i)}) < 1 - \epsilon$.
If $\psi^{(i)}$ violates hyperedge $e$ of $G$,
then $\sigma^{(i)}$ may not satisfy at least one clause of $\phi_e$.
Consequently, $\sigma^{(i)}$ must violate more than $\epsilon |E|$ clauses of $\phi$ in total, and we obtain
\begin{align}
\begin{aligned}
    \val_\phi(\sigmaS)
    \leq \val_\phi(\sigma^{(i)})
    < \frac{m - \epsilon |E|}{m} 
    \underbrace{\leq}_{\text{use } m \leq W^q\cdot 2^{qW}|E|} \frac{m - \frac{\epsilon}{W^q \cdot 2^{qW}} m}{m}
    = 1 - \frac{\epsilon}{W^q \cdot 2^{qW}},
\end{aligned}
\end{align}
which completes the proof.
\end{claimproof}

In the proof of \cref{lem:EkSAT-E3SAT},
we use an established Karp reduction from \prb{$k$-SAT} to \prb{$3$-SAT},
previously used by
{Gopalan, Kolaitis, Maneva, and Papadimitriou}~\cite{gopalan2009connectivity} in the context of reconfiguration.

\begin{claimproof}[Proof of \cref{lem:EkSAT-E3SAT}]
Our reduction is equivalent to that due to
{Gopalan, Kolaitis, Maneva, and Papadimitriou}~\cite[Lemma 3.5]{gopalan2009connectivity}.
Let $(\phi, \sigma_\sss, \sigma_\ttt)$ be an instance of \prb{Maxmin E$k$-SAT Reconfiguration}, 
where $\phi$ is an E$k$-CNF formula consisting of $m$ clauses $C_1, \ldots, C_m$ over $n$ variables $V$ and
$\psi_\sss$ and $\psi_\ttt$ satisfy $\phi$.
Starting from an empty CNF formula $\phi'$,
for each clause $C_j = (\ell_1 \vee \cdots \vee \ell_k)$ of $\phi$,
we introduce $k-3$ new variables $z^j_1, z^j_2, \ldots, z^j_{k-3}$ and
add the following $k-2$ clauses to $\phi'$:
\begin{align}
\label{eq:kSAT-3SAT-trick}
    (\ell_1 \vee \ell_2 \vee z^j_1) \wedge 
    (\ell_3 \vee \bar{z^j_1} \vee z^j_2) \wedge 
    \cdots \wedge
    (\ell_{k-2} \vee \bar{z^j_{k-4}} \vee z^j_{k-3}) \wedge 
    (\ell_{k-1} \vee \ell_k \vee \bar{z^j_{k-3}}).
\end{align}
Observe that a truth assignment makes all clauses of \cref{eq:kSAT-3SAT-trick} satisfied
if and only if it satisfies $C_j$.
Given a satisfying truth assignment $\sigma$ for $\phi$,
consider the following truth assignment $\sigma'$ for $\phi'$:
$\sigma'(x) \triangleq \sigma(x)$ for each variable $x \in V$, and
$\sigma'(z^j_i)$ for each clause $C_j = (\ell_1 \vee \ldots \vee \ell_k)$ is $ \TTT$ if $i \leq i^*-2$ and $\FFF$ if $i \geq i^*-1$, where
$\ell_{i^*}$ evaluates to $\TTT$ by $\sigma$.
Obviously, $\sigma'$ satisfies $\phi'$. 
Constructing $\sigma'_\sss$ from $\sigma_\sss$ and
$\sigma'_\ttt$ from $\sigma_\ttt$
according to this procedure,
we obtain an instance $(\phi', \sigma'_\sss, \sigma'_\ttt)$ of \prb{Maxmin E$3$-SAT Reconfiguration},
which completes the reduction.
Note that $\phi'$ has $(k-2)m$ clauses, and
each variable of $\phi'$ appears in at most $\max\{B, 2\}$ clauses of $\phi'$ if
each variable of $\phi$ appears in at most $B$ clauses of $\phi$.

Since the completeness follows from \cite[Lemma 3.5]{gopalan2009connectivity},
we prove the soundness; i.e.,
$\val_\phi(\sigma_\sss \reco \sigma_\ttt) < 1-\epsilon$ implies
$\val_{\phi'}(\sigma'_\sss \reco \sigma'_\ttt) < 1-\frac{\epsilon}{k-2}$.
Let
$\sigmaS' = \langle \sigma'^{(0)} = \sigma'_\sss, \ldots, \sigma'^{(\ell)} = \sigma'_\ttt \rangle$
be any reconfiguration sequence for $(\phi', \sigma'_\sss, \sigma'_\ttt)$.
Construct then a sequence of assignments,
$\sigmaS = \langle \sigma^{(i)} \rangle_{0 \leq i \leq \ell} $,
such that each $\sigma^{(i)}$ is simply the restriction of $\sigma'^{(i)}$ onto $V$.
Since $\sigmaS$ is a valid reconfiguration sequence for $(\phi, \sigma_\sss, \sigma_\ttt)$,
we have $\val_\phi(\sigmaS) < 1- \epsilon$; in particular,
there exists some $\sigma^{(i)}$ such that $\val_\phi(\sigma^{(i)}) < 1-\epsilon$.
If $\sigma^{(i)}$ violates clause $C_j$, then
$\sigma'^{(i)}$ may not satisfy at least one clause in \cref{eq:kSAT-3SAT-trick}.
Consequently, $\sigma'^{(i)}$ must violate more than $\epsilon m$ clauses of $\phi'$ in total, and we obtain
\begin{align}
    \val_{\phi'}(\sigmaS')
    \leq \val_{\phi'}(\sigma'^{(i)})
    < \frac{(k-2)m - \epsilon m}{(k-2)m}
    = 1- \frac{\epsilon}{k-2},
\end{align}
which completes the proof.
\end{claimproof}

Subsequently, we reduce \prb{Maxmin E$3$-SAT Reconfiguration} to \prb{Maxmin BCSP$_3$ Reconfiguration} in a gap-preserving manner,
whose proof uses the \emph{place encoding} due to {J{\"{a}}rvisalo and Niemel{\"{a}}}~\cite{jarvisalo2004compact}.
\begin{lemma}\label{lem:E3SAT-BCSP}
\begin{leftbar}
For every $\epsilon \in (0,1)$,
there exists a gap-preserving reduction from
\prb{Gap$_{1,1-\epsilon}$ E$3$-SAT Reconfiguration}
to
\prb{Gap$_{1,1-\frac{\epsilon}{3}}$ BCSP$_3$ Reconfiguration}.
Moreover, if the number of occurrences of each variable in the former problem is $B$,
then the maximum degree of the constraint graph in the latter problem is bounded by $\max\{B, 3\}$.
\end{leftbar}
\end{lemma}
\begin{proof}
We first describe a gap-preserving reduction from
\prb{Maxmin E$3$-SAT Reconfiguration}
to
\prb{Maxmin BCSP$_3$ Reconfiguration}.
Let $(\phi, \sigma_\sss, \sigma_\ttt)$ be an instance of
\prb{Maxmin E$3$-SAT Reconfiguration}, where
$\phi$ is an E$3$-CNF formula consisting of $m$ clauses $C_1, \ldots, C_m$ over $n$ variables $x_1, \ldots, x_n$, and
$\sigma_\sss$ and $\sigma_\ttt$ satisfy $\phi$.
Using the place encoding due to {J{\"{a}}rvisalo and Niemel{\"{a}}}~\cite{jarvisalo2004compact},
we construct a binary constraint graph $G = (V,E,\Sigma,\Pi)$ as follows.
The underlying graph of $G$ is a \emph{bipartite graph} with a bipartition
$(\{x_1, \ldots, x_n\}, \{C_1, \ldots, C_m\})$, and
there is an edge between variable $x_i$ and clause $C_j$ in $E$
if $x_i$ or $\bar{x_i}$ appears in $C_j$.
For the sake of notation,
we use $\Sigma_v$ to denote the alphabet assigned to vertex $v \in V$;
we  write
$\Sigma_{x_i} \triangleq \{\TTT, \FFF\}$ for each variable $x_i$, and
$\Sigma_{C_j} \triangleq \{\ell_1, \ell_2, \ell_3\}$ for each clause $C_j = (\ell_1 \vee \ell_2 \vee \ell_3)$.
For each edge $(x_i, C_j) \in E$ with $C_j = (\ell_1 \vee \ell_2 \vee \ell_3)$,
the constraint $\pi_{(x_i, C_j)} \subset \Sigma_{x_i} \times \Sigma_{C_j}$ 
is defined as follows:
\begin{align}
\pi_{(x_i, C_j)} \triangleq
\begin{cases}
    (\Sigma_{x_i} \times \Sigma_{C_j}) \setminus \{(\FFF, x_i)\}
    & \text{if } x_i \text{ appears in } C_j, \\
    (\Sigma_{x_i} \times \Sigma_{C_j}) \setminus \{(\TTT, \bar{x_i})\}
    & \text{if } \bar{x_i} \text{ appears in } C_j.
\end{cases}
\end{align}
Intuitively, 
for an assignment $\psi \colon V \to \Sigma$,
$\psi(x_i)$ claims the truth value assigned to $x_i$, and
$\psi(C_j)$ specifies which literal should evaluate to $\TTT$.
Given a satisfying truth assignment $\sigma$ for $\phi$,
consider the following assignment
$\psi_{\sigma}$ for $G$:
$\psi_{\sigma}(x_i) \triangleq \sigma(x_i)$ for each variable $x_i$, and
$\psi_{\sigma}(C_j) \triangleq \ell_i$ for each clause $C_j$, where
$\ell_i$ appears in $C_j$ and evaluates to $\TTT$ by $\sigma$.\footnote{
Such $\ell_i$ always exists as $\sigma$ satisfies $C_j$.
}
Obviously, $\psi_{\sigma}$ satisfies $G$.
Constructing $\psi_\sss$ from $\sigma_\sss$ and $\psi_\ttt$ from $\sigma_\ttt$
according to this procedure,
we obtain an instance $(G, \psi_\sss, \psi_\ttt)$ of \prb{Maxmin BCSP$_3$ Reconfiguration},
which completes the reduction.
Note that $|V| = n+m$, $|E| = 3m$, and the maximum degree of $G$ is $\max\{B, 3\}$.

We first prove the completeness; i.e.,
$\val_\phi(\sigma_\sss \reco \sigma_\ttt) = 1$ implies $\val_G(\psi_\sss \reco \psi_\ttt) = 1$.
It suffices to consider the case that $\sigma_\sss$ and $\sigma_\ttt$ differ in exactly one variable, say, $x_i$.
Without loss of generality, we can assume that $\sigma_\sss(x_i) = \TTT$ and $\sigma_\ttt(x_i) = \FFF$.
Since both $\sigma_\sss$ and $\sigma_\ttt$ satisfy $\phi$,
for each clause $C_j$ including $x_i$ or $\bar{x_i}$,
there must be a literal $\ell^j$ that is neither $x_i$ nor $\bar{x_i}$ 
and evaluates to $\TTT$ by both $\sigma_\sss$ and $\sigma_\ttt$.
Consider now the following transformation from
$\psi_\sss$ to $\psi_\ttt$:
\begin{itembox}[l]{\textbf{Reconfiguration from $\psi_\sss$ to $\psi_\ttt$}}
\begin{algorithmic}[1]
    \For{\textbf{each} clause $C_j$ including $x_i$ or $\bar{x_i}$}
        \State change the value of $C_j$ from $\psi_\sss(C_j)$ to the aforementioned literal $\ell^j$.
    \EndFor
    \State change the value of $x_i$ from $\TTT$ to $\FFF$.
    \For{\textbf{each} $C_j$ including $x_i$ or $\bar{x_i}$}
        \State change the value of $C_j$ from $\ell^j$ to $\psi_\ttt(C_j)$. 
    \EndFor
\end{algorithmic}
\end{itembox}
Observe easily that every intermediate assignment satisfies $G$; i.e., $\val_G(\psi_\sss \reco \psi_\ttt) = 1$, as desired.

We then prove the soundness; i.e., $\val_\phi(\sigma_\sss \reco \sigma_\ttt) < 1-\epsilon$ implies $\val_G(\psi_\sss \reco \psi_\ttt) < 1-\frac{\epsilon}{3}$.
Let $\psiS = \langle \psi^{(0)} = \psi_\sss, \ldots, \psi^{(\ell)} = \psi_\ttt \rangle$ be
any reconfiguration sequence for $(G, \psi_\sss, \psi_\ttt)$.
Construct then a sequence of truth assignments,
$\sigmaS = \langle \sigma^{(i)} \rangle_{0 \leq i \leq  \ell}$,
such that
each $\sigma^{(i)}$ is simply the restriction of $\psi^{(i)}$ onto the variables of $\phi$.
Since $\sigmaS$ is a valid reconfiguration sequence for $(\phi, \sigma_\sss, \sigma_\ttt)$,
we have $\val_\phi(\sigmaS) < 1-\epsilon$;
in particular, there exists some $\sigma^{(i)}$ such that $\val_\phi(\sigma^{(i)}) < 1-\epsilon$.
If $\sigma^{(i)}$ does not satisfy clause $C_j$, then
$\psi^{(i)}$ violates at least one edge incident to $C_j$ regardless of the assignment to clauses.
Consequently, $\psi^{(i)}$ must violate more than $\epsilon m$ edges of $G$ in total, and we obtain
\begin{align}
    \val_G(\psiS) \leq \val_G(\psi^{(i)}) < \frac{|E| - \epsilon m}{|E|} = 1 - \frac{\epsilon}{3},
\end{align}
which completes the proof.
\end{proof}

\subsection{Degree Reduction of {\normalfont \prb{Maxmin BCSP Reconfiguration}}}
\label{subsec:SAT-CSP:BCSP-degreduce}
We now present a gap-preserving reduction from
\prb{Maxmin BCSP Reconfiguration} to itself of \emph{bounded degree}.
This is the most technical step in this paper.

\begin{lemma}\label{lem:BCSP-deg_reduce}
\begin{leftbar}
For every $\epsilon \in (0,1)$,
there exists a gap-preserving reduction from
\prb{Gap$_{1,1-\epsilon}$ BCSP$_3$ Reconfiguration} to
\prb{Gap$_{1, 1-\bar{\epsilon}}$ BCSP$_6(\Delta)$ Reconfiguration},
where $\bar{\epsilon} \in (0,1)$ and $\Delta \in \bbN$ are some computable functions dependent only on the value of $\epsilon$.
In particular, the constraint graph in the latter problem has bounded degree.
\end{leftbar}
\end{lemma}

\paragraph{Expander Graphs.}
Before proceeding to the details of our reduction, we
introduce concepts related to \emph{expander graphs}.
\begin{definition}\label{def:expander}
\begin{leftbar}
For every $n \in \bbN$, $d \in \bbN$, and $\lambda > 0$,
an \emph{$(n,d,\lambda)$-expander graph} is a $d$-regular graph $G$ on $n$ vertices such that
$\max\{ \lambda_2(G), |\lambda_n(G)| \} \leq \lambda < d$, where
$\lambda_i(G)$ is the \nth{$i$} largest (real-valued) eigenvalue of the adjacency matrix of $G$.
\end{leftbar}
\end{definition}
An $(n,d,\lambda)$-expander graph is called \emph{Ramanujan} if $\lambda \leq 2\sqrt{d-1}$.
There exists an \emph{explicit construction} (i.e., a polynomial-time algorithm)
for near-Ramanujan graphs.
\begin{theorem}[Explicit construction of near-Ramanujan graphs \cite{mohanty2021explicit,alon2021explicit}]
\label{thm:explicit-expander}
\begin{leftbar}
For every $d \geq 3$, $\epsilon>0$, and all sufficiently large $n \geq n_0(d,\epsilon)$, where $nd$ is even,
there is a deterministic $n^{\bigO(1)}$-time algorithm that outputs
an $(n,d,\lambda)$-expander graph with $\lambda \leq 2\sqrt{d-1} + \epsilon$.
\end{leftbar}
\end{theorem}

In this paper, we rely only on the special case of $\epsilon = 2\sqrt{d} - 2 \sqrt{d-1}$ so that
$\lambda \leq 2\sqrt{d}$;
thus, we let $n_0(d) \triangleq n_0(d, 2\sqrt{d} - 2\sqrt{d-1})$.
We can assume $n_0(\cdot)$ to be computable as $2\sqrt{d} - 2\sqrt{d-1} \geq \frac{1}{\sqrt{d}}$.
The crucial property of expander graphs that we use in the proof of \cref{lem:BCSP-deg_reduce}
is the following expander mixing lemma \cite{alon1988explicit}.
\begin{lemma}[Expander mixing lemma; e.g., {Alon and Chung}~\cite{alon1988explicit}]
\label{lem:expander-mixing}
\begin{leftbar}
Let $G$ be an $(n,d,\lambda)$-expander graph.
Then, for any two sets $S$ and $T$ of vertices,
it holds that
\begin{align}
    \left|e(S,T) -  \frac{d|S| \cdot |T|}{n} \right| \leq \lambda \sqrt{|S|\cdot|T|},
\end{align}
where $e(S,T)$ counts the number of edges between $S$ and $T$.
\end{leftbar}
\end{lemma}
This lemma states that 
$e(S,T)$ of an expander graph $G$ is concentrated around
its expectation if $G$ were a \emph{random} $d$-regular graph.
The use of near-Ramanujan graphs enables us to
make an additive error (i.e., $\lambda \sqrt{|S|\cdot|T|}$) acceptably small.

\paragraph{Reduction.}
Our gap-preserving reduction is now presented,
which \emph{does} depend on $\epsilon$.
Redefine $\epsilon \leftarrow \lceil \frac{1}{\epsilon} \rceil^{-1}$
so that $\frac{1}{\epsilon}$ is a positive integer,
which does not increase the value of $\epsilon$;
i.e., $\val_G(\psi_\sss \reco \psi_\ttt) < 1-\epsilon $ implies
$\val_G(\psi_\sss \reco \psi_\ttt) < 1-\lceil \frac{1}{\epsilon} \rceil^{-1}$.
Let $ (G, \psi_\sss, \psi_\ttt) $ be an instance of
\prb{Gap$_{1,1-\epsilon}$ BCSP$_3$ Reconfiguration},
where is $G=(V,E, \Sigma, \Pi = (\pi_e)_{e \in E})$ is a binary constraint graph with $|\Sigma|=3$, and
$\psi_\sss$ and $\psi_\ttt$ satisfy $G$.
For the sake of notation, we denote $\Sigma \triangleq \{\aaa, \bbb, \ccc\}$.
We then create a new instance $(G', \psi'_\sss, \psi'_\ttt)$
of \prb{Maxmin BCSP$_6$ Reconfiguration}, which
turns out to meet the requirement of completeness and soundness.
The ingredients of constraint graph $G'=(V',E', \Sigma', \Pi'=(\pi'_{e'})_{e' \in E'})$ is defined as follows:
\begin{description}
\item[Vertex set:] For each vertex $v$ of $V$, let
\begin{align}
    \cloud(v) \triangleq \Bigl\{ (v,e) \Bigm| e \in E \text{ is incident to } v \Bigr\}.
\end{align}
Define $V' \triangleq \bigcup_{v \in V} \cloud(v)$.
\item[Edge set:] For each vertex $v$ of $V$,
let $X_v$ be a $(d_G(v), d_0, \lambda)$-expander graph on $\cloud(v)$ using \cref{thm:explicit-expander} if $d_G(v) \geq n_0(d_0)$, or
a complete graph on $\cloud(v)$ if $d_G(v) < n_0(d_0)$.
Here, $\lambda \leq 2\sqrt{d_0}$ and
$d_0 = \Theta(\epsilon^{-2})$,
whose precise value will be determined later.
Define
\begin{align}
    E' \triangleq \bigcup_{v \in V} E(X_v)
    \cup \Bigl\{ \left((v, e), (w, e)\right) \in V' \times V' \Bigm| e = (v,w) \in E \Bigr\}.
\end{align}
\item[Alphabet:]
Apply the \emph{alphabet squaring trick} to define
\begin{align}
    \Sigma' \triangleq
    \Bigl\{ \{\aaa\}, \{\bbb\}, \{\ccc\}, \{\aaa, \bbb\}, \{\bbb, \ccc\}, \{\ccc, \aaa\} \Bigr\}.
\end{align}
By abuse of notation,
we write each value of $\Sigma'$ as if it were an element
(e.g., $\aaa\bbb \in \Sigma'$, $\aaa \subset \aaa\bbb$, and $\bbb \not\subseteq \ccc\aaa$).
\item[Constraints:] The constraint $\pi'_{e'} \subseteq \Sigma'^{e'}$ for each edge $e' \in E'$ is defined as follows:
\begin{itemize}
\item If $e' \in E(X_v)$ for some $v \in V$ (i.e., $e'$ is an intra-cloud edge), define\footnote{
\cref{eq:pi-intra} can be expanded as
$
\pi'_{e'} = \{
(\aaa,\aaa),(\bbb,\bbb),(\ccc,\ccc),\allowbreak
(\aaa\bbb,\aaa),(\aaa\bbb,\bbb),\allowbreak
(\bbb\ccc,\bbb),(\bbb\ccc,\ccc),\allowbreak
(\ccc\aaa,\ccc),\allowbreak(\ccc\aaa,\aaa),\allowbreak
(\aaa,\aaa\bbb),(\bbb,\aaa\bbb),\allowbreak
(\bbb,\bbb\ccc),(\ccc,\bbb\ccc),\allowbreak
(\ccc,\ccc\aaa),(\aaa,\ccc\aaa),\allowbreak
(\aaa\bbb,\aaa\bbb),(\bbb\ccc,\bbb\ccc),(\ccc\aaa,\ccc\aaa)
\}
$.
}
\begin{align}\label{eq:pi-intra}
    \pi'_{e'}
    \triangleq \Bigl\{ (\alpha, \beta) \in \Sigma' \times \Sigma' \Bigm|
    \alpha \subseteq \beta \text{ or } \beta \subseteq \alpha
    \Bigr\}.
\end{align}
\item If $e' = ((v,e), (w,e))$ such that $e = (v,w) \in E$ (i.e., $e'$ is an inter-cloud edge), define
\begin{align}\label{eq:pi-inter}
    \pi'_{e'}
    \triangleq \Bigl\{ (\alpha, \beta) \in \Sigma' \times \Sigma' \Bigm|
    \alpha \times \beta \subseteq \pi_{e}
    \Bigr\}.
\end{align}
\end{itemize}
\end{description}
Although the underlying graph $(V', E')$ is the same as that in \cite{dinur2007pcp} (except for the use of \cref{thm:explicit-expander}),
the definitions of $\Sigma'$ and $\Pi'$ are somewhat different
owing to the alphabet squaring trick.
Use of this trick is essential to achieve the perfect completeness.
Intuitively,
having vertex $v' \in V'$ be $\psi(v') = \aaa\bbb$ represents that 
$v'$ has values $\aaa$ \emph{and} $\bbb$ simultaneously; e.g.,
if $\psi'(v') = \aaa\bbb$ and $\psi'(w') = \ccc$ for some
$v' \in \cloud(v)$ and $w' \in \cloud(w)$ with $v \neq w$, then
$\psi'$ satisfies $\pi'_{(v', w')}$ if both $(\aaa,\bbb)$ and $(\aaa,\ccc)$ are found in $\pi_{(v,w)}$
because of \cref{eq:pi-inter}.
Construct two assignments
$\psi'_\sss \colon V' \to \Sigma'$ from $\psi_\sss$ and
$\psi'_\ttt \colon V' \to \Sigma'$ from $\psi_\ttt$ such that
$\psi'_\sss(v,e) \triangleq \{\psi_\sss(v)\}$ and
$\psi'_\ttt(v,e) \triangleq \{\psi_\ttt(v)\}$ for all $(v,e) \in V'$.
Observe that both $\psi'_\sss$ and $\psi'_\ttt$ satisfy $G'$, thereby completing the reduction.
Note that
$|V'| = 2|E|$,
$|E'| \leq n_0(d_0)\cdot |E|$, 
$|\Sigma'| = 6$, and
the maximum degree of $G'$ is $\Delta \leq n_0(d_0)$, which is constant for fixed $\epsilon$.

Using an example illustrated in \cref{fig:example},
we demonstrate that our reduction may map a \NO instance of \prb{BCSP Reconfiguration} to a \YES instance; namely,
$\val_G(\psi_\sss \reco \psi_\ttt) < 1$ does \emph{not} imply
$\val_{G'}(\psi'_\sss \reco \psi'_\ttt) < 1$.
In particular,
it is neither
a Karp reduction of \prb{BCSP Reconfiguration} nor
a PTAS reduction of \prb{Maxmin BCSP Reconfiguration}.
This fact renders the proof of the soundness nontrivial.

\begin{example}\label{ex:BCSP}
\begin{leftbar}
We construct a constraint graph $G = (V,E,\Sigma,\Pi=(\pi_e)_{e \in E})$ such that
$V \triangleq \{ w,v,x,y, \allowbreak z_1, \ldots, z_n \}$ for some large integer $n$,
$E \triangleq \{ (w,v), (v,x), (x,y), (v,z_1), \ldots, (v,z_n) \}$,
$\Sigma \triangleq \{\aaa, \bbb, \ccc\}$, and
each $\pi_e$ is defined as follows:
\begin{align}
\begin{aligned}
    \pi_{(w,v)} & \triangleq \{(\aaa,\aaa)\}, \\
    \pi_{(v,x)} & \triangleq \{(\aaa,\aaa), (\bbb,\aaa), (\bbb,\bbb), (\bbb,\ccc), (\aaa,\ccc)\}, \\
    \pi_{(x,y)} & \triangleq \{(\aaa,\aaa), (\bbb,\aaa), (\bbb,\bbb), (\ccc,\bbb), (\ccc,\ccc)\}, \\
    \pi_{(v,z_1)} & = \cdots = \pi_{(v,z_n)} \triangleq \Sigma \times \Sigma.
\end{aligned}
\end{align}
Define $\psi_\sss, \psi_\ttt \colon V \to \Sigma$
as $\psi_\sss(u) \triangleq \aaa$ for all $u \in V$,
$\psi_\ttt(x) = \psi_\ttt(y) \triangleq \ccc$, and
$\psi_\ttt(u) \triangleq \aaa$ for all other $u$.
Then, it is \emph{impossible} to transform $\psi_\sss$ into $\psi_\ttt$ without any constraint violation:
As the values of $w$ and $v$ cannot change from $\aaa$,
we can only change the value of $x$ to $\ccc$, violating $(x,y)$.
In particular, $\val_G(\psi_\sss \reco \psi_\ttt) < 1$.

Consider applying our reduction to $v$ \emph{only} for the sake of simplicity.
Create $\cloud(v) \triangleq \{ v_w, v_x, v_{z_1}, \ldots, v_{z_n} \}$
with the shorthand notation $v_u \triangleq (v,(v,u))$,
and let $X_v$ be an expander graph on $\cloud(v)$.
We then construct a new constraint graph $G' = (V',E',\Sigma',\Pi'=(\pi'_e)_{e \in E'})$, where
$V' \triangleq \{w,x,y,z_1, \ldots, z_n\} \cup \cloud(v)$,
$E' \triangleq E(X_v) \cup \{(w,v_w), (v_x,x), (x,y), (v_{z_1},z_1), \ldots, (v_{z_n},z_n)\}$,
$\Sigma' \triangleq \{\aaa,\bbb,\ccc,\aaa\bbb,\bbb\ccc,\ccc\aaa\}$, and
each constraint $\pi'_e$ is defined according to \cref{eq:pi-intra,eq:pi-inter}.
Construct $\psi'_\sss, \psi'_\ttt \colon V' \to \Sigma'$ from $\psi_\sss, \psi_\ttt$
according to the procedure described above.
Suppose now ``by chance'' $(v_w, v_x) \not\in E(X_v)$.
The crucial observation is that 
we can assign $\aaa$ to $v_w$, $\bbb$ to $v_x$, and $\aaa\bbb$ to $v_{z_1}, \ldots, v_{z_n}$
to do some ``cheating.''
Consequently, $\psi'_\sss$ can be transformed into $\psi'_\ttt$ without sacrificing any constraint:
Assign $\aaa\bbb$ to $v_{z_1}, \ldots, v_{z_n}$ in arbitrary order;
assign $\bbb$ to $v_x$, $x$, and $y$ in this order;
assign $\ccc$ to $x$ and $y$ in this order;
assign $\aaa$ to $v_x$;
assign $\aaa$ to $v_{z_1}, \ldots, v_{z_n}$ in arbitrary order.
In particular, $\val_{G'}(\psi'_\sss \reco \psi'_\ttt) = 1$.
\end{leftbar}
\end{example}

\paragraph{Correctness.}
The proof of the completeness is immediate from the definition of $\Sigma'$ and $\Pi'$.
\begin{lemma}\label{lem:BCSP-deg_reduce:complete}
\begin{leftbar}
If $\val_G(\psi_\sss \reco \psi_\ttt) = 1$, then $\val_{G'}(\psi'_\sss \reco \psi'_\ttt) = 1$.
\end{leftbar}
\end{lemma}
\begin{proof}
It suffices to consider the case that $\psi_\sss$ and $\psi_\ttt$ differ in exactly one vertex, say, $v \in V$.
Let $\alpha \triangleq \psi_\sss(v)$ and $\beta \triangleq \psi_\ttt(v)$.
Note that
$\psi'_\sss(v') = \{\alpha\} \neq \{\beta\} = \psi'_\ttt(v')$ for all $v' \in \cloud(v)$.
On the other hand,
$\psi'_\sss(w') = \{\psi_\sss(w)\} = \{\psi_\ttt(w)\} = \psi'_\ttt(w')$ for all $w' \in \cloud(w)$ with $w \neq v$.
Consider the following transformation $\psiS'$ from $\psi'_\sss$ to $\psi'_\ttt$:
\begin{itembox}[l]{\textbf{Reconfiguration from $\psi'_\sss$ to $\psi'_\ttt$}}
\begin{algorithmic}[1]
        \State change the value of $v'$ in $\cloud(v)$ from $\{\alpha\}$ to $\{\alpha, \beta\}$ one by one.
        \State change the value of $v'$ in $\cloud(v)$ from $\{\alpha, \beta\}$ to $\{\beta\}$ one by one.
\end{algorithmic}
\end{itembox}
In any intermediate step of this transformation,
the set of values that vertices in $\cloud(v)$ have taken is
either $\{\{\alpha\}, \{\alpha, \beta\}\}$, $\{\{\alpha, \beta\}\}$, or $\{ \{\alpha, \beta\}, \{\beta\} \}$;
thus, every assignment of $\psiS'$ satisfies all intra-cloud edges in $E(X_v)$ by \cref{eq:pi-intra}.
Plus, every assignment of $\psiS'$ satisfies all inter-cloud edges $(v',w') \in E$ with $v' \in \cloud(v)$ and $w' \in \cloud(w)$
because
\begin{align}
\begin{aligned}
    (\{\alpha\}, \{\psi_\sss(w)\}) = (\{\alpha\}, \{\psi_\ttt(w)\}) & \in \pi'_{(v', w')}, \\
    (\{\beta\}, \{\psi_\sss(w)\}) = (\{\beta\}, \{\psi_\ttt(w)\}) & \in \pi'_{(v', w')}, \\
    (\{\alpha, \beta\}, \{\psi_\sss(w)\}) = (\{\alpha, \beta\}, \{\psi_\ttt(w)\}) & \in \pi'_{(v', w')},
\end{aligned}
\end{align}
where the last membership relation holds owing to \cref{eq:pi-inter}.
Accordingly, every assignment of $\psiS'$ satisfies $G'$;
i.e., $\val_{G'}(\psiS') = 1$, as desired.
\end{proof}

In the remainder of this subsection, we focus on proving the soundness.
\begin{lemma}\label{lem:BCSP-deg_reduce:soundness}
\begin{leftbar}
If $\val_G(\psi_\sss \reco \psi_\ttt) < 1 - \epsilon$,
then 
$\val_{G'}(\psi'_\sss \reco \psi'_\ttt) < 1 - \bar{\epsilon}$,
where $\bar{\epsilon} = \bar{\epsilon}(\epsilon)$ is some computable function such that
$\bar{\epsilon} \in (0,1)$ if $\epsilon \in (0,1)$.
\end{leftbar}
\end{lemma}
For an assignment $\psi' \colon V' \to \Sigma'$ for $G'$,
let $\PLR(\psi') \colon V \to \Sigma$ denote an assignment for $G$
such that $\PLR(\psi')(v)$ for $v \in V$ is
determined based on the \emph{plurality vote} of $\psi'(v')$ over $v' \in \cloud(v)$; namely,
\begin{align}
    \PLR(\psi')(v) \triangleq \argmax_{\alpha \in \Sigma} \left|\Bigl\{ v' \in \cloud(v) \Bigm| \alpha \in \psi'(v') \Bigr\}\right|,
\end{align}
where ties are arbitrarily broken according to any prefixed ordering over $\Sigma$ (e.g., $\aaa \prec \bbb \prec \ccc$).
Suppose we are given
a reconfiguration sequence $\psiS' = \langle \psi'^{(0)} = \psi'_\sss, \ldots, \psi'^{(\ell)} = \psi'_\ttt \rangle$
for $(G', \psi'_\sss, \psi'_\ttt)$ having the maximum value.
Construct then a sequence of assignments,
$\psiS \triangleq \langle \psi^{(i)} \rangle_{0 \leq i \leq \ell}$,
such that $\psi^{(i)} \triangleq \PLR(\psi'^{(i)})$ for all $i$.
Observe that
$\psiS$ is a valid reconfiguration sequence for $(G, \psi_\sss, \psi_\ttt)$, and
we thus must have $\val_G(\psiS) < 1-\epsilon$;
in particular, there exists some $\psi'^{(i)}$ such that
$\val_G(\PLR(\psi'^{(i)})) = \val_G(\psi^{(i)}) < 1-\epsilon$.
We would like to show that $\val_{G'}(\psi'^{(i)}) < 1-\bar{\epsilon}$
for some constant $\bar{\epsilon} \in (0,1)$ depending only on $\epsilon$.
Hereafter, we denote
$\psi \triangleq \psi^{(i)}$ and $\psi' \triangleq \psi'^{(i)}$ for notational simplicity.

For each vertex $v \in V$,
we define $D_v$ as the set of vertices in $\cloud(v)$ whose values \emph{disagree} with the plurality vote $\psi(v)$; namely,
\begin{align}
D_v \triangleq \Bigl\{ v' \in \cloud(v) \Bigm| \psi(v) \not\in \psi'(v') \Bigr\}.
\end{align}
Consider any edge $e=(v,w) \in E$ violated by $\psi$ (i.e., $(\psi(v), \psi(w)) \not\in \pi_e$), and
let $e' = (v',w') \in E'$ be a unique (inter-cloud) edge such that $v' \in \cloud(v)$ and $w' \in \cloud(w)$.
By definition of $\pi'_{e'}$,
(at least) either of the following conditions must hold:
\begin{description}
    \item[(Condition 1)]
    edge $e'$ is violated by $\psi'$ (i.e., $(\psi'(v'), \psi'(w')) \not\in \pi'_{e'}$), or
    \item[(Condition 2)]
    $\psi(v) \not\in \psi'(v')$ (i.e., $v' \in D_v$) or
    $\psi(w) \not\in \psi'(w')$ (i.e., $w' \in D_w$).
\end{description}
Consequently, 
the number of edges in $E$ violated by $\psi$ is bounded by
the sum of the number of inter-cloud edges in $E'$ violated by $\psi'$ and
the number of vertices in $V'$ who disagree with the plurality vote;
namely,
\begin{align}
\epsilon |E| < \text{(\# inter-cloud edges violated by } \psi' \text{)} + \sum_{v \in V} |D_v|.
\end{align}
Then, one of the two terms on the right-hand side
of the above inequality should be greater than $\frac{\epsilon}{2}|E|$.
If the first term is more than $\frac{\epsilon}{2} |E| $,
then we are done because
\begin{align}\label{eq:BCSP-deg_reduce:ub1}
    \val_{G'}(\psi')
    \leq \frac{|E'| - \text{(\# edges violated by } \psi' \text{)}}{|E'|}
    < 1- \frac{\epsilon}{2} \frac{|E|}{|E'|} \leq 1 - \frac{\epsilon}{2 \cdot n_0(d_0)}.
\end{align}
We now consider the case that
$\sum_{v \in V} |D_v| > \frac{\epsilon}{2} |E|$.
Define $x_v$ for each $v \in V$ as
the fraction of vertices in $\cloud(v)$ who disagree with $\psi(v)$;
namely,
\begin{align}
    x_v \triangleq \frac{|D_v|}{|\cloud(v)|} = \frac{|D_v|}{d_G(v)}.
\end{align}
We also define $\delta \triangleq \frac{\epsilon}{8}$.
We first show that
the total size of $|D_v|$ \emph{conditioned on} $x_v \geq \delta$ is $\Theta(\epsilon|E|)$.

\begin{claim}\label{clm:BCSP-deg_reduce:sum}
\begin{leftbar}
$\displaystyle \sum_{v \in V : x_v \geq \delta} |D_v| > \frac{\epsilon}{4} |E|$,
where $\delta = \frac{\epsilon}{8}$.
\end{leftbar}
\end{claim}
\begin{claimproof}
Note that
\begin{align}
\begin{aligned}
    \sum_{v \in V} |D_v|
    & = \sum_{v: x_v \geq \delta} |D_v| + \sum_{v: x_v < \delta} x_v \cdot d_G(v) \\
    & \leq \sum_{v: x_v \geq \delta} |D_v| + \delta \sum_{v: x_v < \delta} d_G(v)
    \leq \sum_{v: x_v \geq \delta} |D_v| + 2 \delta |E|.
\end{aligned}
\end{align}
Therefore, it holds that
\begin{align}
    \sum_{v: x_v \geq \delta} |D_v|
    \geq \sum_{v \in V} |D_v| - 2 \delta |E|
    \underbrace{>}_{\text{use } \sum_{v \in V} |D_v| > \frac{\epsilon}{2} |E|} \frac{\epsilon}{2} |E| - 2 \delta |E|
    = \frac{\epsilon}{4} |E|,
\end{align}
which completes the proof.
\end{claimproof}
We then discover a pair of disjoint subsets of $\cloud(v)$ for every $v \in V$ such that 
their size is $\Theta(|D_v|)$ and they are mutually conflicting under $\psi'$,
where the fact that $|\Sigma| = 3$ somewhat simplifies the proof by cases.

\begin{observation}\label{clm:BCSP-deg_reduce:ST}
\begin{leftbar}
For each vertex $v$ of $V$,
there exists a pair of disjoint subsets $S$ and $T$ of $\cloud(v)$ such that
$|S| \geq \frac{|D_v|}{3}$, 
$|T| \geq \frac{|D_v|}{3}$, and 
$\psi'$ violates all constraints between $S$ and $T$.
\end{leftbar}
\end{observation}
\begin{proof}
Without loss of generality, we can assume that $\psi(v) = \aaa$.
For each value $\alpha \in \Sigma'$,
let $n_\alpha$ denote the number of vertices in $\cloud(v)$ whose value is exactly $\alpha$; namely,
\begin{align}
    n_\alpha \triangleq \left|\Bigl\{ v' \in \cloud(v) \Bigm| \psi'(v') = \alpha \Bigr\}\right|.
\end{align}
By definition of $D_v$, we have $n_{\bbb} + n_{\ccc} + n_{\bbb\ccc} = |D_v|$.
Since one of $n_{\bbb}$, $n_{\ccc}$, or $n_{\bbb\ccc}$ must be at least $\frac{|D_v|}{3}$,
we have the following three cases to consider:

\begin{description}
\item[(Case 1)] If $n_\bbb \geq \frac{|D_v|}{3}$:
By construction of $\psi$ by the plurality vote on $\psi'$, we have
\begin{align}
\begin{aligned}
    \underbrace{n_{\aaa} + n_{\aaa\bbb} + n_{\ccc\aaa}}_{\text{\# vertices contributing to } \aaa}
    & \geq \underbrace{n_{\bbb} + n_{\aaa\bbb} + n_{\bbb\ccc}}_{\text{\# vertices contributing to } \bbb} \\
    \implies n_{\aaa} + n_{\ccc\aaa} & \geq n_{\bbb} + n_{\bbb\ccc} \geq n_{\bbb} \geq \frac{|D_v|}{3}.
\end{aligned}
\end{align}
Therefore, we let
$S \triangleq \{v' \in \cloud(v) \mid \psi'(v') \mbox{ is } \bbb \}$
and 
$T \triangleq \{v' \in \cloud(v) \mid \psi'(v') \mbox{ is } \aaa \mbox{ or } \ccc\aaa \}$
to ensure that
$|S|, |T| \geq \frac{|D_v|}{3} $ and
every intra-cloud edge between $S$ and $T$ is violated by $\psi'$ owing to \cref{eq:pi-intra}.
\item[(Case 2)] If $n_{\ccc} \geq \frac{|D_v|}{3}$:
Similarly, we have
\begin{align}
\begin{aligned}
    \underbrace{n_{\aaa} + n_{\aaa\bbb} + n_{\ccc\aaa}}_{\text{\# vertices contributing to } \aaa}
    & \geq \underbrace{n_{\ccc} + n_{\ccc\aaa} + n_{\bbb\ccc}}_{\text{\# vertices contributing to } \ccc} \\
    \implies n_{\aaa} + n_{\aaa\bbb} & \geq n_{\ccc} + n_{\bbb\ccc} \geq n_{\ccc} \geq \frac{|D_v|}{3}.
\end{aligned}
\end{align}
Thus, we let
$S \triangleq \{v' \in \cloud(v) \mid \psi'(v') \mbox{ is } \ccc\}$
and 
$T \triangleq \{v' \in \cloud(v) \mid \psi'(v') \mbox{ is } \aaa \mbox{ or } \aaa\bbb\}$
to have that
$|S|, |T| \geq \frac{|D_v|}{3}$ and
all intra-cloud edges between $S$ and $T$ are unsatisfied.
\item[(Case 3)] If $n_{\bbb\ccc} \geq \frac{|D_v|}{3}$:
Observe that
\begin{align}
\begin{aligned}
    \underbrace{n_{\aaa} + n_{\aaa\bbb} + n_{\ccc\aaa}}_{\text{\# vertices contributing to } \aaa}
    & \geq \underbrace{n_{\bbb} + n_{\aaa\bbb} + n_{\bbb\ccc}}_{\text{\# vertices contributing to } \bbb} \\
    & \geq n_{\bbb\ccc} \geq \frac{|D_v|}{3}.
\end{aligned}
\end{align}
Letting
$S \triangleq \{v' \in \cloud(v) \mid \psi'(v') \mbox{ is } \bbb\ccc\}$
and 
$T \triangleq \{v' \in \cloud(v) \mid \psi'(v') \mbox{ is } \aaa, \aaa\bbb, \mbox{or } \ccc\aaa\}$ is sufficient.
\end{description}
The above case analysis finishes the proof.
\end{proof}

Consider a vertex $v \in V$ such that $x_v \geq \delta$; that is,
at least $\delta$-fraction of vertices in $\cloud(v)$ disagree with $\psi(v)$.
Letting $S$ and $T$ be two disjoint subsets of $\cloud(v)$ obtained by \cref{clm:BCSP-deg_reduce:ST},
we wish to bound the number of edges between $S$ and $T$ (i.e., $e(S,T)$) using the expander mixing lemma.
Hereafter, we determine the value of $d_0$ by
$d_0 \triangleq \left(\frac{12}{\delta}\right)^2 = \frac{9216}{\epsilon^2}$,
which is a positive \emph{even} integer (so that \cref{thm:explicit-expander} is applicable) and depends only on the value of $\epsilon$.
Suppose first $d_G(v) \geq n_0(d_0)$; i.e., $X_v$ is an expander.

\begin{lemma}\label{lem:BCSP-deg_reduce:eST}
\begin{leftbar}
For a vertex $v$ of $V$ such that $x_v \geq \delta$ and $d_G(v) \geq n_0(d_0)$,
let $S$ and $T$ be a pair of disjoint subsets of $\cloud(v)$ obtained by \cref{clm:BCSP-deg_reduce:ST}.
Then,
$e(S,T) \geq \frac{8}{\delta} |D_v|$.
\end{leftbar}
\end{lemma}
\begin{proof}
Recall that $X_v$ is a $(d_G(v), d_0, \lambda)$-expander graph,
where $\lambda \leq 2 \sqrt{d_0}$.
By applying the expander mixing lemma on $S$ and $T$,
we obtain
\begin{align}
    e(S,T)
    \geq \frac{d_0 |S|\cdot|T|}{d_G(v)} - \lambda \sqrt{|S|\cdot |T|}
    \geq
    \underbrace{\frac{|S| \cdot |T|}{d_G(v)} \left(\frac{12}{\delta}\right)^2 - \frac{2 \cdot 12}{\delta} \sqrt{|S|\cdot |T|}}_{= \underline{e}(S,T)}.
\end{align}
Consider $\underline{e}(S,T)$ as a quadratic polynomial in $\sqrt{|S|\cdot |T|}$.
Setting the partial derivative of $\underline{e}(S,T)$ by $\sqrt{|S|\cdot |T|}$ equal to $0$,
we obtain
\begin{align}
\begin{aligned}
    & \frac{\partial}{\partial \sqrt{|S|\cdot |T|}} \underline{e}(S,T) =
    \frac{2 \sqrt{|S|\cdot |T|}}{d_v} \left(\frac{12}{\delta} \right)^2 - \frac{2 \cdot 12}{\delta} = 0 \\
    \implies & \sqrt{|S|\cdot |T|} = \frac{\delta}{12}d_v.
\end{aligned}
\end{align}
Therefore, $\underline{e}(S,T)$ is monotonically increasing in $\sqrt{|S|\cdot |T|}$ when
$\sqrt{|S|\cdot |T|} > \frac{\delta}{12}d_G(v)$.
Observing that
$\sqrt{|S|\cdot |T|} \geq \frac{\delta}{3}d_G(v)$ since
$|S| \geq \frac{x_v}{3}d_G(v)$, $|T| \geq \frac{x_v}{3}d_G(v)$, and $x_v \geq \delta$ by assumption,
we derive
\begin{align}
\begin{aligned}
    e(S,T) & \geq \underline{e}(S,T) \geq \frac{1}{d_G(v)} \left(\frac{x_v \cdot d_G(v)}{3} \right)^2 \left(\frac{12}{\delta}\right)^2 - \frac{2 \cdot 12}{\delta} \frac{x_v \cdot d_G(v)}{3} \\
    & \underbrace{\geq}_{\text{use } x_v \geq \delta} \frac{1}{d_G(v)} \left(\frac{x_v \cdot d_G(v)}{3} \right) \left(\frac{\delta \cdot d_G(v)}{3} \right) \left(\frac{12}{\delta} \right)^2
    - \frac{2\cdot 12}{\delta} \frac{x_v \cdot d_G(v)}{3} \\
    & = \frac{16}{\delta} x_v \cdot d_G(v) - \frac{8}{\delta} x_v \cdot d_G(v)
    = \frac{8}{\delta} |D_v|. \qedhere
\end{aligned}
\end{align}
\end{proof}

Suppose then $d_G(v) < n_0(d_0)$.
Since $X_v$ forms a complete graph over $d_G(v)$ vertices,
$e(S,T)$ is exactly equal to $|S|\cdot |T|$, which is evaluated as
\begin{align}
\label{eq:BCSP-deg_reduce:ST-small}
    e(S,T) = |S|\cdot |T|
    \underbrace{\geq}_{\text{\cref{clm:BCSP-deg_reduce:ST}}} \left(\frac{|D_v|}{3}\right)^2
    = \frac{x_v \cdot d_G(v)}{9} |D_v|
    \underbrace{\geq}_{\text{use } d_G(v) \geq 1 \text{ and } x_v \geq \delta} \frac{\delta}{9} |D_v|.
\end{align}
By \cref{lem:BCSP-deg_reduce:eST,eq:BCSP-deg_reduce:ST-small},
for every vertex $v \in V$ such that $x_v \geq \delta$,
the number of violated intra-cloud edges within $X_v$ is 
at least $\min\{\frac{8}{\delta}, \frac{\delta}{9} \} |D_v| \geq \frac{\delta}{9} |D_v|$.
Simple calculation using \cref{clm:BCSP-deg_reduce:sum} bounds the total number of intra-cloud edges violated by $\psi'$ from below as
\begin{align}
    \sum_{v \in V} (\text{\# edges in } X_v \text{ violated by } \psi')
    \geq \sum_{v: x_v \geq \delta} \frac{\delta}{9} |D_v|
    \underbrace{>}_{\text{\cref{clm:BCSP-deg_reduce:sum}}}
    \frac{\epsilon}{72} \frac{\epsilon}{4} |E|
    \geq \frac{\epsilon^2 \cdot |E'|}{288 \cdot n_0(d_0)}. \label{eq:BCSP-deg_reduce:ub2}
\end{align}
Consequently, from \cref{eq:BCSP-deg_reduce:ub1,eq:BCSP-deg_reduce:ub2}, we conclude that
\begin{align}
    \val_{G'}(\psiS') \leq \val_{G'}(\psi')
    < \max\left\{1-\frac{\epsilon}{2 \cdot n_0(d_0)}, 1-\frac{\epsilon^2}{288 \cdot n_0(d_0)}\right\}
    = 1-\frac{\epsilon^2}{288 \cdot n_0\left(\frac{9216}{\epsilon^2}\right)}.
\end{align}
Setting $\bar{\epsilon} \triangleq \frac{\epsilon^2}{288 \cdot n_0\left(\frac{9216}{\epsilon^2}\right)}$ accomplishes
the proof of \cref{lem:BCSP-deg_reduce:soundness}
and thus \cref{lem:BCSP-deg_reduce}. \qed

\subsection{Putting It Together}
\label{subsec:SAT-CSP:together}

We are now ready to finish the proof of \cref{thm:E3SAT}.
\begin{proof}[Proof of \cref{thm:E3SAT}]
By \cref{lem:qCSP-E3SAT,lem:E3SAT-BCSP},
\prb{Gap$_{1,1-\epsilon}$ BCSP$_3$ Reconfiguration} is \PSPACE-hard
for some $\epsilon \in (0,1)$ under \cref{hyp:RIH}.
Thus, under the same hypothesis,
\prb{Gap$_{1,1-\bar{\epsilon}}$ BCSP$_6(\Delta)$ Reconfiguration}
is \PSPACE-hard for some $\bar{\epsilon} \in (0,1)$ and $\Delta \in \bbN$
depending only on $\epsilon$ as guaranteed by \cref{lem:BCSP-deg_reduce}.
Since the maximum degree of input constraint graphs is bounded by $\Delta$,
we further apply \cref{lem:qCSP-E3SAT} to conclude that
\prb{Gap$_{1,1-\epsilon'}$ E$3$-SAT$(B)$ Reconfiguration} is \PSPACE-hard
under the hypothesis
for some $\epsilon' \in (0,1)$ and $B \in \bbN$ depending solely on $\epsilon$,
which accomplishes the proof.
\end{proof}

\section{Applications}\label{sec:app}
Here, we apply \cref{thm:E3SAT} to 
devise conditional \PSPACE-hardness of approximation for
\prb{Nondeterministic Constraint Logic},
popular reconfiguration problems on graphs, and
\prb{$2$-SAT Reconfiguration}.

\subsection{Optimization Variant of {\normalfont \prb{Nondeterministic Constraint Logic}}}
\label{subsec:app:NCL}
We review \prb{Nondeterministic Constraint Logic} invented by
{Hearn and Demaine}~\cite{hearn2005pspace,hearn2009games}.
An \emph{\scAND/\scOR graph} is defined as an undirected graph $G=(V,E)$, where
each link of $E$ is colored \emph{red} or \emph{blue} and has weight $1$ or $2$, respectively, and
each node of $V$ is one of the following two types:\footnote{
We refer to vertices and edges of an \scAND/\scOR graph as
\emph{nodes} and \emph{links} to distinguish from those of a standard graph.
}
\begin{itemize}
    \item \emph{\scAND node}, which has two incident red links and one incident blue link, or
    \item \emph{\scOR node}, which has three incident blue links.
\end{itemize}
Hence, every \emph{\scAND/\scOR graph} is $3$-regular.
An orientation (i.e., an assignment of direction to each link) of $G$ 
\emph{satisfies} a particular node of $G$ if the total weight of its incoming links is at least $2$, and 
\emph{satisfies} $G$ if all nodes are satisfied.
\scAND and \scOR nodes are designed to behave like the corresponding logical gates:
the blue link of an \scAND node can be directed outward if and only if
both two red links are directed inward;
a particular blue link of an \scOR node can be directed outward if and only if
at least one of the other two blue links is directed inward.
Thus, a direction of each link can be considered a \emph{signal}.
In the \prb{Nondeterministic Constraint Logic} problem,
for an \scAND/\scOR graph $G$ and its two satisfying orientations $O_\sss$ and $O_\ttt$,
we are asked if $O_\sss$ can be transformed into $O_\ttt$ by
a sequence of link reversals
while ensuring that every intermediate orientation satisfies $G$.\footnote{
A variant of \prb{Nondeterministic Constraint Logic},
called \emph{configuration-to-edge} \cite{hearn2005pspace},
requires to decide if a specified link can be eventually reversed by a sequence of link reversals.
From a point of view of approximability, this definition does not seem to make much sense.
}

We now formulate an optimization variant of \prb{Nondeterministic Constraint Logic},
which affords to use an orientation that does \emph{not} satisfy some nodes.
Once more, we define $\val_G(\cdot)$ for \scAND/\scOR graph $G$ analogously:
Let $\val_G(O)$ denote the fraction of nodes satisfied by orientation $O$, let
\begin{align}
    \val_G(\scrO) \triangleq \min_{O^{(i)} \in \scrO} \val_G(O^{(i)})
\end{align}
for reconfiguration sequence of orientations, $\scrO = \langle O^{(i)} \rangle_{0 \leq i \leq \ell}$, and let
\begin{align}
    \val_G(O_\sss \reco O_\ttt) \triangleq \max_{\scrO = \langle O_\sss, \ldots, O_\ttt \rangle} \val_G(\scrO)
\end{align}
for two orientations $O_\sss$ and $O_\ttt$.
Then, for a pair of orientations $O_\sss$ and $O_\ttt$ of $G$,
\prb{Maxmin Nondeterministic Constraint Logic} requires to maximize $\val_G(\scrO)$
subject to $\scrO = \langle O_\sss, \ldots, O_\ttt \rangle$, and
for every $0 \leq s \leq c \leq 1$,
\prb{Gap$_{c,s}$ Nondeterministic Constraint Logic} requests to distinguish whether
$\val_G(O_\sss \reco O_\ttt) \geq c$ or $\val_G(O_\sss \reco O_\ttt) < s$.
We demonstrate that
RIH implies \PSPACE-hardness of approximation for \prb{Maxmin Nondeterministic Constraint Logic}.

\begin{proposition}\label{thm:NCL}
\begin{leftbar}
For every $B \in \bbN$ and $\epsilon \in (0,1)$,
there exists a gap-preserving reduction from
\prb{Gap$_{1, 1-\epsilon}$ E$3$-SAT$(B)$ Reconfiguration}
to
\prb{Gap$_{1, 1-\Theta(\frac{\epsilon}{B})}$ Nondeterministic Constraint Logic}.
\end{leftbar}
\end{proposition}

Our proof makes a modification to the CNF network \cite{hearn2005pspace,hearn2009games}.
To this end, we refer to special nodes that can be simulated by an \scAND/\scOR subgraph, including
\scCHOICE, \scREDBLUE, \scFANOUT nodes, and free-edge terminators,
which are described blow; see also {Hearn and Demaine}~\cite{hearn2005pspace,hearn2009games} for more details.
\begin{itemize}
\item \emph{\scCHOICE node}:
    This node has three red links and
    is satisfied if at least two links are directed inward;
    i.e., only one link may be directed outward.
    A particular constant-size \scAND/\scOR subgraph can emulate a \scCHOICE node, wherein
    some nodes would be unsatisfied whenever two or more red links are directed outward.
\item \emph{\scREDBLUE node}:
    This is a degree-two node incident to one red edge and one blue link, which
    acts as transferring a signal between them; i.e.,
    one link may be directed outward if and only if
    the other is directed inward.
    A specific constant-size \scAND/\scOR subgraph can simulate a \scREDBLUE node, wherein
    some nodes become unsatisfied whenever
    both red and blue links are directed outward.
\item \emph{\scFANOUT node}:
    This node is equivalent to an \scAND node from a different interpretation:
    two red links may be directed outward if and only if the blue link is directed inward.
    Accordingly, a \scFANOUT node plays a role in \emph{splitting} a signal.
\item \emph{Free-edge terminator}:
    This is an \scAND/\scOR subgraph of constant size used to connect the loose end of a link.
    The connected link is free in a sense that it can be directed inward or outward.
\end{itemize}

\paragraph{Reduction.}
Given an instance $(\phi,\sigma_\sss,\sigma_\ttt)$ of \prb{Maxmin E$3$-SAT$(B)$ Reconfiguration},
where $\phi$ is an E$3$-CNF formula consisting of $m$ clauses $C_1, \ldots, C_m$ over $n$ variables $x_1, \ldots, x_n$, and
$\sigma_\sss$ and $\sigma_\ttt$ satisfy $\phi$,
we construct an \scAND/\scOR graph $G_\phi$ as follows.
For each variable $x_i$ of $\phi$, we create a \scCHOICE node, denoted $v_{x_i}$, called a \emph{variable node}.
Of the three red links incident to $v_{x_i}$, 
one is connected to a free-edge terminator, whereas
the other two are labeled ``$x_i$'' and ``$\bar{x_i}$.''
Thus, either of the links $x_i$ or $\bar{x_i}$ can be directed outward without sacrificing $v_{x_i}$.
For each clause $C_j$ of $\phi$, we create an \scOR node, denoted $v_{C_j}$, called a \emph{clause node}.
The output signals of variable nodes' links are sent toward the corresponding clause nodes.
Specifically,
if literal $\ell$ appears in multiple clauses of $\phi$,
we first make a desired number of copied signals of link $\ell$ using 
\scREDBLUE and \scFANOUT nodes;
if $\ell$ does not appear in any clause, we connect link $\ell$ to a free-edge terminator.
Then, for each clause $C_j = (\ell_1 \vee \ell_2 \vee \ell_3)$ of $\phi$,
the clause node $v_{C_j}$ is connected to three links corresponding to
the (copied) signals of $\ell_1, \ell_2, \ell_3$.
This completes the construction of $G_\phi$.
See \cref{fig:NCL} for an example.

Observe that $G_\phi$ is satisfiable if and only if $\phi$ is satisfiable \cite{hearn2005pspace,hearn2009games}.
Given a satisfying truth assignment $\sigma$ for $\phi$,
we can construct a satisfying orientation $O_\sigma$ of $G_\phi$: 
the trick is that
if literal $x_i$ or $\bar{x_i}$ appearing in clause $C_j$ evaluates to $\TTT$ by $\sigma$ ,
we can safely orient \emph{every} link on the unique path between $v_{x_i}$ and $v_{C_j}$ toward $v_{C_j}$.
Constructing $O_\sss$ from $\sigma_\sss$ and $O_\ttt$ from $\sigma_\ttt$ according to this procedure,
we obtain an instance $(G, O_\sss, O_\ttt)$ of \prb{Maxmin Nondeterministic Constraint Logic}, which completes the reduction.
The proof of the correctness shown below relies on the fact that for fixed $B \in \bbN$,
the number of nodes in $G_\phi$ is proportional to the number of variable nodes $n$ as well as that of clause nodes $m$.

\begin{figure}[t]
    \centering
    \scalebox{0.7}{\input{NCL}}
    \caption{
    An \scAND/\scOR graph $G_\phi$ corresponding to an E$3$-CNF formula
    $\phi = (w\vee x\vee y) \wedge (w \vee \bar{x} \vee z) \wedge (x \vee \bar{y} \vee z)$,
    taken and modified from \cite[Figure 5.1]{hearn2009games}.
    Here,
    thicker blue links have weight $2$,
    thinner red links have weight $1$, and
    the square node denotes a free-edge terminator.
    The orientation of $G_\phi$ shown above is given by $O_{\psi_\sss}$ such that
    $\psi_\sss(w,x,y,z) = (\FFF,\TTT,\TTT,\TTT)$.
    If $\psi_\ttt$ is defined as
    $\psi_\ttt(w,x,y,z) = (\FFF,\FFF,\TTT,\TTT)$,
    we can transform $O_{\psi_\sss}$ into $O_{\psi_\ttt}$;
    in particular, all links in the subtree rooted at $x$, denoted the gray area, can be made directed downward.
    }
    \label{fig:NCL}
\end{figure}
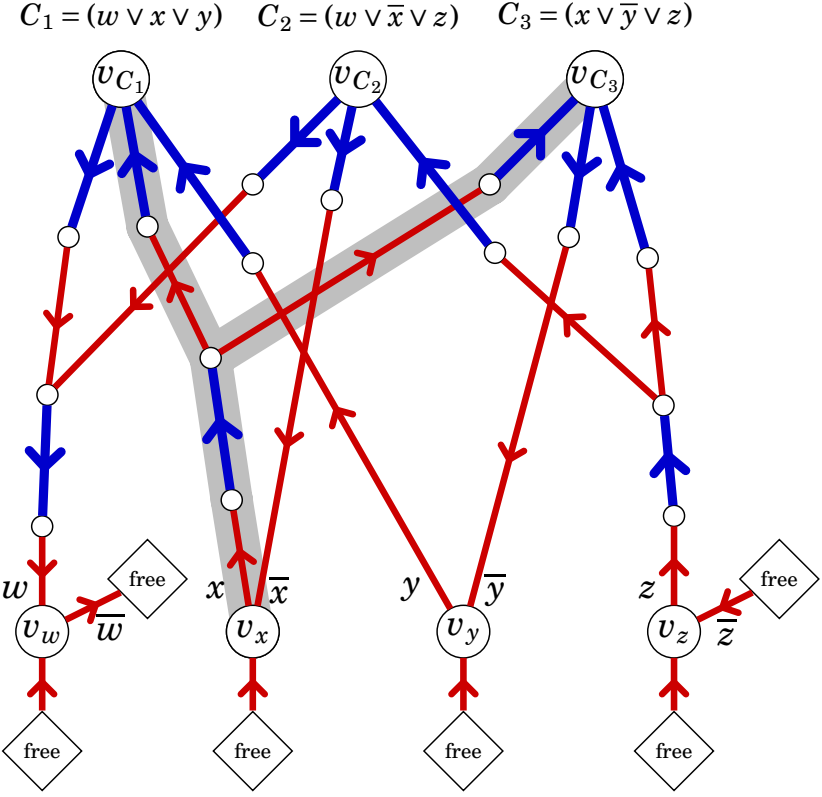

\begin{proof}[Proof of \cref{thm:NCL}]
We begin with a few remarks on the construction of $G_\phi$.
For each clause $C_j$ that includes literal $x_i$ or $\bar{x_i}$,
there is a \emph{unique path} between $v_{x_i}$ and $v_{C_j}$ without passing through any other variable or clause node,
which takes the following form:
\begin{align*}
    & \text{Output signal of a variable node } v_{x_i} \\
    & \rightarrow \text{ a \scREDBLUE node} \\
    & \rightarrow \text{ any number of (a \scFANOUT node } \rightarrow \text{ a \scREDBLUE node)} \\
    & \rightarrow \text{ a clause node } v_{C_j}.
\end{align*}
Therefore, every node of $G_\phi$ excepting variable and clause nodes
is uniquely associated with a particular literal $\ell$ of $\phi$.
Hereafter, the \emph{subtree rooted at literal $\ell$} is defined as
a subgraph of $G_\phi$ induced by
the unique paths between the corresponding variable node and
clause nodes $v_{C_j}$ for $C_j$ including $\ell$ (see also \cref{fig:NCL}).

We first prove the completeness; i.e.,
$\val_\phi(\sigma_\sss \reco \sigma_\ttt) = 1$ implies $\val_{G_\phi}(O_\sss \reco O_\ttt) = 1$.
It suffices to consider the case that
$\sigma_\sss$ and $\sigma_\ttt$ differ in exactly one variable, say, $x_i$.
Without loss of generality, we can assume that $\sigma_\sss(x_i) = \TTT$ and $\sigma_\ttt(x_i) = \FFF$; i.e.,
link $x_i$ is directed outward (resp.~inward) in $O_\sss$ (resp.~$O_\ttt$).
Since both $\sigma_\sss$ and $\sigma_\ttt$ satisfy $\phi$,
for each clause $C_j$ including either $x_i$ or $\bar{x_i}$, 
at least one of the remaining two literals of $C_j$ evaluates to $\TTT$ by both $\sigma_\sss$ and $\sigma_\ttt$.
Furthermore, in the subtree rooted at such a literal, every link is directed toward the leaves (i.e., clause nodes)
in both $O_\sss$ and $O_\ttt$.
By this observation, we can safely transform $O_\sss$ into $O_\ttt$ as follows, as desired:
\begin{itembox}[l]{\textbf{Reconfiguration from $O_\sss$ to $O_\ttt$}}
\begin{algorithmic}[1]
    \State orient every link in the subtree rooted at $x_i$ toward $v_{x_i}$, along the leaves (i.e., clause nodes including $x_i$) to the root $v_{x_i}$.
    \LComment{both links $x_i$ and $\bar{x_i}$ become directed inward.}
    \State orient every link in the subtree rooted at $x_i$ toward $v_{C_j}$ for all $C_j$ including $\bar{x_i}$,
    along the root $v_{x_i}$ to the leaves.
\end{algorithmic}
\end{itembox}

We then prove the soundness; i.e.,
$\val_\phi(\sigma_\sss \reco \sigma_\ttt) < 1-\epsilon$ implies $\val_{G_\phi}(O_\sss \reco O_\ttt) < 1-\Theta(\frac{\epsilon}{B})$.
Let $\scrO = \langle O^{(0)} = O_\sss, \ldots, O^{(\ell)} = O_\ttt \rangle$ be
any reconfiguration sequence for $(G_\phi, O_\sss, O_\ttt)$.
Construct then a sequence of truth assignments, $\sigmaS = \langle \sigma^{(i)} \rangle_{0 \leq i \leq \ell}$,
such that
each $\sigma^{(i)}(x_j)$ for variable $x_j$
is $\TTT$ if ``link $x_j$ is directed outward from $v_{x_j}$ and link $\bar{x_j}$ is directed inward to $v_{x_j}$,'' and 
is $\FFF$ otherwise.
Since $\sigmaS$ is a valid reconfiguration sequence for $(\phi, \sigma_\sss, \sigma_\ttt)$,
we have $\val_\phi(\sigmaS) < 1-\epsilon$; in particular,
there exists some $\sigma^{(i)}$ such that $\val_\phi(\sigma^{(i)}) < 1-\epsilon$.
Unfortunately,
the number of \emph{clause nodes} satisfied by $O^{(i)}$ may be not less than $m(1-\epsilon)$ 
because other nodes may be violated in lieu of clause nodes (e.g., both  $x_i$ and $\bar{x_i}$ may be directed outward).
Thus, we compare $O^{(i)}$ with an orientation $O_{\sigma^{(i)}}$ constructed from $\sigma^{(i)}$ by the 
procedure described in the reduction paragraph.
Note that $O_{\sigma^{(i)}}$ satisfies every non-clause node, while
more than $\epsilon m$ clause nodes are unsatisfied.
Transforming $O_{\sigma^{(i)}}$ into $O^{(i)}$ by reversing the directions of conflicting links one by one,
we can see that each time a non-clause node becomes unsatisfied owing to a link reversal,
we would be able to make at most $B$ clause nodes satisfied.
Consequently, we derive
\begin{align}
\begin{aligned}
    & \underbrace{\epsilon m}_{\substack{\text{\# clause nodes} \\ \text{violated by } O_{\sigma^{(i)}}}}
    - B \cdot (\text{\# non-clause nodes violated by } O^{(i)})
    < (\text{\# clause nodes violated by } O^{(i)}) \\
    & \implies
    (\text{\# nodes violated by } O^{(i)}) > \frac{\epsilon}{B} m \\
    & \implies
    \val_{G_\phi}(\scrO) \leq \val_{G_\phi}(O^{(i)}) < \frac{|V(G_\phi)| - \frac{\epsilon}{B}m}{|V(G_\phi)|} = 1-\Theta\left(\frac{\epsilon}{B}\right),
\end{aligned}
\end{align}
where we used that $|V(G_\phi)| = \Theta(m+n) = \Theta(m)$,
completing the proof.
\end{proof}

\subsection{Reconfiguration Problems on Graphs}
\label{sec:app:graph}

\paragraph{Independent Set Reconfiguration and Clique Reconfiguration.}
We first consider \prb{Independent Set Reconfiguration} and its optimization variant.
Denote by $\alpha(G)$ the size of maximum independent sets of a graph $G$.
Two independent sets of $G$ are \emph{adjacent} if 
one is obtained from the other by adding or removing a single vertex of $G$; i.e.,
their symmetric difference has size $1$.
Such a model of reconfiguration is called \emph{token addition and removal} \cite{ito2011complexity}.\footnote{
We do not consider token jumping \cite{kaminski2012complexity} or token sliding \cite{hearn2005pspace}
since they do not change the size of an independent set.
}
For a pair of independent sets $I_\sss$ and $I_\ttt$ of a graph $G$,
\prb{Independent Set Reconfiguration} asks if
there is a reconfiguration sequence 
from $I_\sss$ to $I_\ttt$
made up of independent sets only of size at least
$\min\{ |I_\sss|, |I_\ttt| \} - 1$.
For a reconfiguration sequence of independent sets of $G$,
denoted $\scrI = \langle I^{(i)} \rangle_{0 \leq i \leq \ell}$,
let
\begin{align}
    \val_G(\scrI) \triangleq \min_{I^{(i)} \in \scrI} \frac{|I^{(i)}|}{\alpha(G)-1}.
\end{align}
Here, division by $\alpha(G)-1$ is derived from the nature that
reconfiguration from $I_\sss$ to $I_\ttt$ entails a vertex removal
whenever $|I_\sss|=|I_\ttt|=\alpha(G)$ and $I_\sss \neq I_\ttt$.
Then, for a pair of independent sets $I_\sss$ and $I_\ttt$ of $G$,
\prb{Maxmin Independent Set Reconfiguration} requires to maximize 
$\val_G(\scrI)$ subject to
$\scrI = \langle I_\sss, \ldots, I_\ttt \rangle$,
which is known to be \NP-hard to approximate within any constant factor \cite{ito2011complexity}.
Subsequently,
let $\val_G(I_\sss \reco I_\ttt)$ denote
the maximum value of $\val_G(\scrI)$
over all possible reconfiguration sequences $\scrI$
from $I_\sss$ to $I_\ttt$; namely,
\begin{align}
\val_G(I_\sss \reco I_\ttt) \triangleq \max_{\scrI = \langle I_\sss, \ldots, I_\ttt \rangle} \val_G(\scrI).
\end{align}
For every $0 \leq s \leq c \leq 1$,
\prb{Gap$_{c,s}$ Independent Set Reconfiguration} requests to distinguish whether
$\val_G(I_\sss \reco I_\ttt) \geq c$ or $\val_G(I_\sss \reco I_\ttt) < s$.
The proof of the following corollary is based on a Karp reduction due to \cite{hearn2005pspace,hearn2009games}.

\begin{corollary}\label{cor:NCL-IS}
\begin{leftbar}
For every $\epsilon \in (0,1)$,
there exists a gap-preserving reduction from 
\prb{Gap$_{1,1-\epsilon}$ Nondeterministic Constraint Logic} to
\prb{Gap$_{1,1-\Theta(\epsilon)}$ Independent Set Reconfiguration}.
In particular, \prb{Maxmin Independent Set Reconfiguration}
is \PSPACE-hard to approximate within constant factor under \cref{hyp:RIH}.
\end{leftbar}
\end{corollary}

As an immediate corollary,
\prb{Maxmin Clique Reconfiguration} is \PSPACE-hard to approximate under RIH.
\begin{corollary}\label{cor:IS-C}
\begin{leftbar}
For every $\epsilon \in (0,1)$,
there exists a gap-preserving reduction from 
\prb{Gap$_{1,1-\epsilon}$ Nondeterministic Constraint Logic}
to 
\prb{Gap$_{1,1-\Theta(\epsilon)}$ Clique Reconfiguration}.
In particular, \prb{Maxmin Clique Reconfiguration}
\PSPACE-hard to approximate within constant factor under \cref{hyp:RIH}.
\end{leftbar}
\end{corollary}

\begin{proof}[Proof of \cref{cor:NCL-IS}]
We show that a Karp reduction from
\prb{Nondeterministic Constraint Logic} to \prb{Independent Set Reconfiguration}
due to \cite{hearn2005pspace,hearn2009games} is indeed gap preserving.
Let $(G, O_\sss, O_\ttt)$ be an instance of \prb{Maxmin Nondeterministic Constraint Logic}, where
$G=(V,E)$ is an \scAND/\scOR graph made up of $n_{\scAND}$ \scAND nodes and $n_{\scOR}$ \scOR nodes, and
$O_\sss$ and $O_\ttt$ satisfy $G$.
Construct a graph $G' = (V',E')$ 
by replacing
each \scAND node by an \scAND gadget and
each \scOR node by an \scOR gadget due to \cite{hearn2005pspace,hearn2009games},
which are drawn in \cref{fig:gadget}.
According to an interpretation of \scAND/\scOR graphs due to {Bonsma and Cereceda}~\cite{bonsma2009finding},
$G'$ consists of \emph{token edges}, each of which is a copy of $K_2$ across the border of gadgets, and
\emph{token triangles}, each of which is a copy of $K_3$ appearing only in an \scOR gadget.
Observe easily that
the number of token edges is $n_e = \frac{3}{2} (n_{\scAND} + n_{\scOR})$,
the number of token triangles is $n_t = n_{\scOR}$, and
thus $|V'| = 2n_e + 3n_t = 3n_{\scAND} + 6n_{\scOR}$.
Given a satisfying orientation $O$ of $G$,
we can construct a maximum independent set $I_O$ of $G'$ as follows \cite{hearn2005pspace,hearn2009games}:
Of each token edge $e$ across the gadgets corresponding to nodes $v$ and $w$,
we choose $e$'s endpoint on $w$'s side (resp.~$v$'s side)
if link $(v,w)$ is directed toward $v$ (resp.~$w$) under $O$;
afterwards, we can select one vertex from each token triangle since
at least one blue link of the respective \scOR node must be directed inward.
Since $I_O$ includes one vertex from each token edge/triangle,
it holds that $|I_O| = \alpha(G') = n_e + n_t = \frac{3}{2}n_{\scAND} + \frac{5}{2}n_{\scOR}$.
Constructing $I_\sss$ from $O_\sss$ and $I_\ttt$ from $O_\ttt$ according to this procedure,
we obtain an instance $(G', I_\sss, I_\ttt)$ of \prb{Maxmin Independent Set Reconfiguration},
which completes the reduction.

Since the completeness follows from \cite{hearn2005pspace,hearn2009games},
we prove (the contraposition of) the soundness; i.e.,
$\val_{G'}(I_\sss \reco I_\ttt) \geq 1-\epsilon$ implies $\val_G(O_\sss \reco O_\ttt) \geq 1-6\epsilon$
for $\epsilon \in (0, \frac{1}{6})$ and 
sufficiently large $n_{\scAND} + n_{\scOR}$.
Suppose we have a reconfiguration sequence $\scrI = \langle I^{(i)} \rangle_{0 \leq i \leq \ell}$ for $(G',I_\sss,I_\ttt)$ such that
$\val_{G'}(\scrI) \geq 1- \epsilon$.
Construct then a sequence of orientations, $\scrO = \langle O^{(i)} \rangle_{0 \leq i \leq \ell}$,
where each $O^{(i)}$ is defined as follows:
for each token edge $e$ across the gadgets corresponding to nodes $v$ and $w$,
link $(v,w)$ is made directed toward $v$ if $I^{(i)}$
includes $e$'s endpoint on $w$'s side, and
is made directed toward $w$ otherwise.
By definition,
if $I^{(i)}$ does not intersect with
a particular token edge/triangle
(in particular, $|I^{(i)}| < \alpha(G')$),
$O^{(i)}$ may not satisfy nodes of $G$ corresponding to the gadgets overlapping with that token edge/triangle.
On the other hand, because each token edge/triangle intersects up to two gadgets, 
at most $2(\alpha(G') - |I^{(i)}|)$ nodes may be
unsatisfied.
Consequently, using that
$\min_{I^{(i)} \in \scrI} |I^{(i)}| \geq (1-\epsilon)(\alpha(G)-1)$, we get
\begin{align}
\begin{aligned}
    \val_G(\scrO)
    & \geq \min_{O^{(i)} \in \scrO} \frac{|V| - \text{(\# nodes violated by } O^{(i)}\text{)}}{|V|} \\
    & \geq \frac{|V| - 2(\alpha(G') - \min_{I^{(i)} \in \scrI} |I^{(i)}|)}{|V|} \\
    & \geq \frac{|V| - 2\epsilon \cdot \alpha(G') - 2(1-\epsilon)}{|V|} \\
    & = \frac{(n_{\scAND} + n_{\scOR}) - 2\epsilon \cdot (\frac{3}{2}n_{\scAND} + \frac{5}{2}n_{\scOR}) - 2(1-\epsilon)}{n_{\scAND} + n_{\scOR}} \\
    & = \frac{(1-3\epsilon)n_{\scAND} + (1-5\epsilon)n_{\scOR} -2(1-\epsilon)}{n_{\scAND} + n_{\scOR}} \\
    & \geq 1-6\epsilon
    \quad\quad\text{ for all } n_{\scAND}+n_{\scOR} \geq \frac{2}{\epsilon},
\end{aligned}
\end{align}
which completes the proof.
\end{proof}

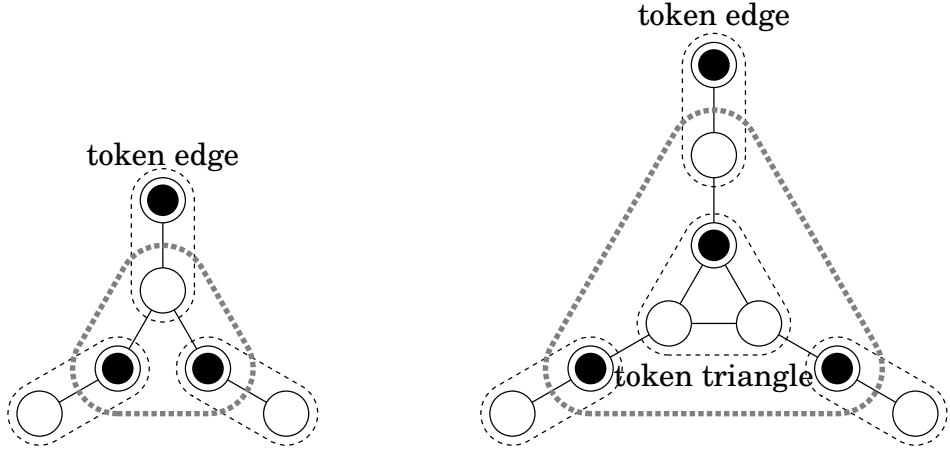
\begin{figure}[tbp]
    \centering
    \null\hfill
    \scalebox{0.6}{\input{AND}}
    \hfill
    \scalebox{0.6}{\input{OR}}
    \hfill\null
    \caption{
        \scAND gadget (left) and \scOR gadget (right), taken and modified from \cite[Figure 9.14]{hearn2009games}.
        Dashed black lines correspond to token edges or token triangles.
        Dotted gray lines represent gadget borders.
    }
    \label{fig:gadget}
\end{figure}

\paragraph{Vertex Cover Reconfiguration.}
We conclude this section with \prb{Minmax Vertex Cover Reconfiguration},
which is known to be $2$-factor approximable \cite{ito2011complexity}.
Denote by $\beta(G)$ the size of minimum vertex covers of a graph $G$.
Just like in \prb{Independent Set Reconfiguration},
we adopt the token addition and removal model to define the adjacency relation; that is,
two vertex covers are \emph{adjacent} if their symmetric difference has size $1$.
For a pair of vertex covers $C_\sss$ and $C_\ttt$ of a graph $G$,
\prb{Vertex Cover Reconfiguration} asks if
there is a reconfiguration sequence from $C_\sss$ to $C_\ttt$
made up of vertex covers of size at most
$\max\{|C_\sss|, |C_\ttt|\}+1$.
We further use analogous notations to those in \prb{Maxmin Independent Set Reconfiguration}:
Let
\begin{align}
    \val_G(\scrC) \triangleq \max_{C^{(i)} \in \scrC} \frac{|C^{(i)}|}{\beta(G)+1}
\end{align}
for a reconfiguration sequence of vertex covers of $G$,
$\scrC = \langle C^{(i)} \rangle_{0 \leq i \leq \ell}$,
and let
\begin{align}
    \val_G(C_\sss \reco C_\ttt) \triangleq
    \min_{\scrC = \langle C_\sss, \ldots, C_\ttt \rangle}
    \val_G(\scrC)
\end{align}
for a pair of vertex covers $C_\sss$ and $C_\ttt$ of $G$.
Then,
for a pair of vertex covers $C_\sss$ and $C_\ttt$ of $G$,
\prb{Minmax Vertex Cover Reconfiguration}
requires to minimize $\val_G(\scrC)$ subject to $\scrC = \langle C_\sss, \ldots, C_\ttt \rangle$, whereas
for every $1 \leq c \leq s$,
\prb{Gap$_{c,s}$ Vertex Cover Reconfiguration} requests to distinguish whether
$\val_G(C_\sss \reco C_\ttt) \leq c$ or $\val_G(C_\sss \reco C_\ttt) > s$.
The proof of the following result uses a gap-preserving reduction from
\prb{Maxmin Independent Set Reconfiguration} obtained from \cref{cor:NCL-IS}.

\begin{corollary}\label{cor:IS-VC}
\begin{leftbar}
For every $\epsilon \in (0,1)$,
there exists a gap-preserving reduction from
\prb{Gap$_{1,1-\epsilon}$ Nondeterministic Constraint Logic}
to
\prb{Gap$_{1,1 + \Theta(\epsilon)}$ Vertex Cover Reconfiguration}.
In particular, \prb{Minmax Vertex Cover Reconfiguration}
\PSPACE-hard to approximate within constant factor under \cref{hyp:RIH}.
\end{leftbar}
\end{corollary}
\begin{proof}
We show that a Karp reduction from \prb{Independent Set Reconfiguration} to
\prb{Vertex Cover Reconfiguration} due to \cite{hearn2005pspace,hearn2009games} is indeed gap preserving.
Let $(G, I_\sss, I_\ttt)$ be an instance of \prb{Independent Set Reconfiguration}, where
$G = (V,E)$ is a restricted graph obtained from \cref{cor:NCL-IS} built up with $n_e$ token edges and $n_t$ token triangles
such that $\alpha(G) = n_e + n_t$, and
$I_\sss$ and $I_\ttt$ are maximum independent sets of $G$.
Recall that $|V| = 2n_e + 3n_t$, and thus $\beta(G) = |V| - \alpha(G) = n_e + 2n_t$.
Construct then an instance $(G, C_\sss \triangleq V \setminus I_\sss, C_\ttt \triangleq V \setminus I_\ttt)$ of
\prb{Maxmin Vertex Cover Reconfiguration}.
If there exists a reconfiguration sequence $\scrI = \langle I^{(i)} \rangle_{0 \leq i \leq \ell}$ for $(G, I_\sss, I_\ttt)$ such that $\val_G(\scrI) = 1$,
its complement,
$\scrC = \langle C^{(i)} \rangle_{0 \leq i \leq \ell}$
such that $C^{(i)} \triangleq V \setminus I^{(i)}$ for all $i$,
satisfies
\begin{align}
    \val_G(\scrC)
    = \max_{C^{(i)} \in \scrC} \frac{|C^{(i)}|}{\beta(G) + 1}
    = \frac{|V| - \min_{I^{(i)} \in \scrI} |I^{(i)}|}{\beta(G)+1}
    \leq \frac{|V| - (\alpha(G) - 1)}{|V| - \alpha(G) + 1} = 1,
\end{align}
which finishes the completeness.
Suppose for a reconfiguration sequence $\scrC = \langle C^{(i)}\rangle_{0 \leq i \leq \ell}$ for $(G,C_\sss,C_\ttt)$,
its complement,
$\scrI = \langle I^{(i)} \rangle_{0 \leq i \leq \ell}$
such that $I^{(i)} \triangleq V \setminus C^{(i)}$ for all $i$,
satisfies that $\val_G(\scrI) < 1-\epsilon$.
Since $\min_{I^{(i)} \in \scrI} |I^{(i)}| < (1-\epsilon)(\alpha(G)-1)$,
we get
\begin{align}
\begin{aligned}
    \val_G(\scrC)
    & = \frac{|V| - \min_{I^{(i)} \in \scrI} |I^{(i)}|}{\beta(G)+1} \\
    & > \frac{|V| - (1-\epsilon)(\alpha(G)-1)}{\beta(G) + 1} \\
    & = \frac{(1+\epsilon)n_e + (2+\epsilon)n_t + (1-\epsilon)}{n_e + 2n_t + 1}
    \geq 1 + \frac{\epsilon}{3} \text{\quad for all } n_e + n_t \geq 4,
\end{aligned}
\end{align}
which completes the soundness.
\end{proof}

\subsection{{\normalfont \prb{Maxmin $2$-SAT$(B)$ Reconfiguration}}}
\label{sec:app:2-SAT}

We show that \prb{Maxmin $2$-SAT Reconfiguration} of bounded occurrence
is 
\PSPACE-hard to approximate under RIH as a corollary of
\cref{thm:E3SAT}.
Therefore, we have a simple analogy between
\prb{$2$-SAT} and its reconfiguration version:
One one hand, \prb{$2$-SAT Reconfiguration} \cite{ito2011complexity} as well as \prb{$2$-SAT} are solvable in polynomial time;
on the other hand,
both \prb{Maxmin $2$-SAT Reconfiguration} and \prb{Max $2$-SAT} \cite{hastad2001some} are hard to approximate.
\begin{corollary}\label{cor:E3SAT-2SAT}
\begin{leftbar}
For every $B \in \bbN$ and $\epsilon \in (0,1)$,
there exists a gap-preserving reduction from
\prb{Gap$_{1,1-\epsilon}$ E$3$-SAT$(B)$ Reconfiguration} to
\prb{Gap$_{\frac{7}{10}, \frac{7}{10}-\epsilon}$ $2$-SAT$(4B)$ Reconfiguration}.
In particular, \prb{Maxmin $2$-SAT Reconfiguration} of bounded occurrence is
\PSPACE-hard to approximate within constant factor under \cref{hyp:RIH}.
\end{leftbar}
\end{corollary}

\begin{table}[tbp]
    \centering
    \begin{tabular}{c|cc|cc|cc|cc}
        \toprule
        $\ell_1$ & \multicolumn{2}{c|}{$\FFF$} & \multicolumn{2}{c|}{$\TTT$} & \multicolumn{2}{c|}{$\TTT$} & \multicolumn{2}{c}{$\TTT$} \\
        $\ell_2$ & \multicolumn{2}{c|}{$\FFF$} & \multicolumn{2}{c|}{$\FFF$} & \multicolumn{2}{c|}{$\TTT$} & \multicolumn{2}{c}{$\TTT$} \\
        $\ell_3$ & \multicolumn{2}{c|}{$\FFF$} & \multicolumn{2}{c|}{$\FFF$} & \multicolumn{2}{c|}{$\FFF$} & \multicolumn{2}{c}{$\TTT$} \\
        $z^j$ & $\FFF$ & $\TTT$ & $\FFF$ & $\TTT$ & $\FFF$ & $\TTT$ & $\FFF$ & $\TTT$ \\
        \midrule
        $\bar{\ell_1} \vee \bar{\ell_2}$ &
        $\TTT$ & $\TTT$ & $\TTT$ & $\TTT$ & $\FFF$ & $\FFF$ & $\FFF$ & $\FFF$ \\
        $\bar{\ell_2} \vee \bar{\ell_3}$ &
        $\TTT$ & $\TTT$ & $\TTT$ & $\TTT$ & $\TTT$ & $\TTT$ & $\FFF$ & $\FFF$ \\
        $\bar{\ell_3} \vee \bar{\ell_1}$ &
        $\TTT$ & $\TTT$ & $\TTT$ & $\TTT$ & $\TTT$ & $\TTT$ & $\FFF$ & $\FFF$ \\
        $\ell_1 \vee \bar{z^j}$ &
        $\TTT$ & $\FFF$ & $\TTT$ & $\TTT$ & $\TTT$ & $\TTT$ & $\TTT$ & $\TTT$ \\
        $\ell_2 \vee \bar{z^j}$ &
        $\TTT$ & $\FFF$ & $\TTT$ & $\FFF$ & $\TTT$ & $\TTT$ & $\TTT$ & $\TTT$ \\
        $\ell_3 \vee \bar{z^j}$ &
        $\TTT$ & $\FFF$ & $\TTT$ & $\FFF$ & $\TTT$ & $\FFF$ & $\TTT$ & $\TTT$ \\
        \midrule
         \# satisfied clauses in \cref{eq:2SAT-trick} & $6$ & $4$ & $7$ & $6$ & $7$ & $7$ & $6$ & $7$ \\
        \bottomrule
    \end{tabular}
    \caption{Relation between the truth assignments to $\ell_1, \ell_2, \ell_3, z^j$ and
    the number of satisfied clauses in \cref{eq:2SAT-trick}.}
    \label{tab:2SAT}
\end{table}

\begin{proof}[Proof of \cref{cor:E3SAT-2SAT}]
We first recapitulate a Karp reduction from
\prb{$3$-SAT} to \prb{Max $2$-SAT} due to
{Garey, Johnson, and Stockmeyer}~\cite{garey1976some}.
Let $(\phi, \sigma_\sss, \sigma_\ttt)$ be an instance of \prb{Maxmin E$3$-SAT Reconfiguration}, where
$\phi$ is an E$3$-CNF formula consisting of
$m$ clauses $C_1, \ldots, C_m$ over $n$ variables $x_1, \ldots, x_n$, and
$\sigma_\sss$ and $\sigma_\ttt$ satisfy $\phi$.
Starting with an empty $2$-CNF formula $\phi'$,
for each clause $C_j = (\ell_1 \vee \ell_2 \vee \ell_3)$,
we introduce a new variable $z^j$ and add the following ten clauses to $\phi'$:
\begin{align}\label{eq:2SAT-trick}
    (\ell_1) \wedge (\ell_2) \wedge (\ell_3) \wedge (z^j) \wedge
    (\bar{\ell_1} \vee \bar{\ell_2}) \wedge
    (\bar{\ell_2} \vee \bar{\ell_3}) \wedge
    (\bar{\ell_3} \vee \bar{\ell_1}) \wedge
    (\ell_1 \vee \bar{z^j}) \wedge 
    (\ell_2 \vee \bar{z^j}) \wedge 
    (\ell_3 \vee \bar{z^j})
\end{align}
\cref{tab:2SAT} shows the relation between
the truth assignments to $\ell_1, \ell_2, \ell_3, z^j$ and
the number of clauses satisfied in \cref{eq:2SAT-trick}.
In particular, if $C_j$ is satisfied, then
we can satisfy exactly seven of the ten clauses in \cref{eq:2SAT-trick}
by setting the truth value of $z^j$ appropriately;
otherwise, we can only satisfy at most six clauses.
Given a satisfying truth assignment $\sigma$ for $\phi$,
consider the following truth assignment $\sigma'$ for $\phi'$:
$\sigma'(x_i) \triangleq \sigma(x_i)$ for all $i \in [n]$, and
$\sigma'(z^j)$ for each $j \in [m]$ is $\FFF$ if one or two literals of $C_j$ evaluate to $\TTT$ by $\sigma$, and
is $\TTT$ otherwise
(i.e., if all three literals evaluate to $\TTT$ by $\sigma$).
Observe from \cref{tab:2SAT} that
$\sigma'$ satisfies exactly $\frac{7}{10}$-fraction of clauses of $\phi'$.
Constructing $\sigma'_\sss$ from $\sigma_\sss$ and
$\sigma'_\ttt$ from $\sigma_\ttt$ according to this procedure, 
we obtain an instance $(\phi', \sigma'_\sss, \sigma'_\ttt)$ of
\prb{Maxmin $2$-SAT Reconfiguration}, which completes the reduction.
Note that
$\phi'$ has $10m$ clauses, and
$\val_{\phi'}(\sigma'_\sss) = \val_{\phi'}(\sigma'_\ttt) = \frac{7}{10}$.

We first prove the completeness; i.e., $\val_\phi(\sigma_\sss \reco \sigma_\ttt) = 1$ implies
$\val_{\phi'}(\sigma'_\sss \reco \sigma'_\ttt) = \frac{7}{10}$.
It suffices to consider the case that
$\sigma_\sss$ and $\sigma_\ttt$ differ in one variable, say, $x_i$.
For each clause $C_j$ of $\phi$,
we use $n^j_\sss$ and $n^j_\ttt$ to denote
the number of literals in $C_j$ evaluating to $\TTT$
by $\sigma_\sss$ and $\sigma_\ttt$, respectively.
Then, consider the following transformation from $\sigma'_\sss$ to $\sigma'_\ttt$:
\begin{itembox}[l]{\textbf{Reconfiguration from $\sigma'_\sss$ to $\sigma'_\ttt$}}
\begin{algorithmic}[1]
    \For{\textbf{each} $j \in [m]$}
        \State if $(n^j_\sss, n^j_\ttt) = (2, 3)$, flip the assignment of $z^j$ from $\FFF$ to $\TTT$;
    otherwise, do nothing.
    \EndFor
    \State flip the assignment of $x_i$.
    \For{\textbf{each} $j \in [m]$}
        \State if $(n^j_\sss, n^j_\ttt) = (3, 2)$, flip the assignment of $z^j$ from $\TTT$ to $\FFF$;
        otherwise, do nothing.
    \EndFor
\end{algorithmic}
\end{itembox}
Observe from \cref{tab:2SAT} that 
every intermediate truth assignment satisfies exactly $7m$ clauses; i.e.,
$\val_{\phi'}(\sigma'_\sss \reco \sigma'_\ttt) = \frac{7m}{10m} = \frac{7}{10}$,
as desired.

We then prove the soundness; i.e.,
$\val_\phi(\sigma_\sss \reco \sigma_\ttt) < 1- \epsilon$ implies
$\val_{\phi'}(\sigma'_\sss \reco \sigma'_\ttt) < \frac{7}{10} - \epsilon$.
Let $\sigmaS' = \langle \sigma'^{(0)} = \sigma'_\sss, \ldots, \sigma'^{(\ell)} = \sigma'_\ttt \rangle$
be any reconfiguration sequence for $(\phi', \sigma'_\sss, \sigma'_\ttt)$.
Construct then a sequence of truth assignments,
$\sigmaS = \langle \sigma^{(i)} \rangle_{0 \leq i \leq \ell}$, such that
each $\sigma^{(i)}$ is defined as the restriction of $\sigma'^{(i)}$ onto the variables of $\phi$.
Since $\sigmaS$ is a valid reconfiguration sequence,
we have $\val_\phi(\sigmaS) < 1-\epsilon$;
in particular,
there exists some $\sigma^{(i)} \in \sigmaS$ such that
$\val_\phi(\sigma^{(i)}) < 1-\epsilon$.
If $\sigma^{(i)}$ violates clause $C_j$,
then $\sigma'^{(i)}$ can satisfy at most six clauses in \cref{eq:2SAT-trick}.
Consequently, 
$\sigma'^{(i)}$ satisfies less than
$7 \cdot (1-\epsilon)m + 6 \cdot \epsilon m$ clauses
of $\phi'$, and we derive
\begin{align}
    \val_{\phi'}(\sigmaS') \leq \val_{\phi'}(\sigma')
    < \frac{7 \cdot (1-\epsilon)m + 6 \cdot \epsilon m}{10m}
    = \frac{7}{10} - \epsilon,
\end{align}
thereby completing the proof.
\end{proof}

\section{Conclusions}
We gave a series of gap-preserving reductions to demonstrate
\PSPACE-hardness of approximation for optimization variants of popular reconfiguration problems \emph{assuming} the Reconfiguration Inapproximability Hypothesis (RIH).
An immediate open question is to verify RIH.
One approach is to prove it directly, e.g., by using gap amplification of {Dinur}~\cite{dinur2007pcp}.
Some steps may be more difficult to prove, as we are required to preserve reconfigurability.
Another way entails a reduction from
some problems already known to be \PSPACE-hard to approximate, such as
\prb{True Quantified Boolean Formula} due to {Condon, Feigenbaum, Lund, and Shor}~\cite{condon1995probabilistically}.
We are currently uncertain whether we can ``adapt'' a Karp reduction from 
\prb{True Quantified Boolean Formula} to \prb{Nondeterministic Constraint Logic} \cite{hearn2005pspace,hearn2009games}.

%% file: reductions.tex
\begin{tikzpicture}
    \tikzset{edge/.style={-Latex[round], black, ultra thick}};
    \tikzset{problem/.style={rectangle, rounded corners, thick, draw=black, fill=white, font=\normalsize, align=center, outer sep=0, minimum height=1cm}};
    
    \node[problem, minimum width=3cm] at (0,2.25) (qCSP){\prb{$q$-CSP$_W$ Reconf} \\ \cref{hyp:RIH}};
    \node[problem, minimum width=3cm] at (0,0) (E3SAT){\prb{E$3$-SAT Reconf} \\ \cref{lem:qCSP-E3SAT}};
    \node[problem, minimum width=3cm] at (4,2.25) (BCSP){\prb{BCSP$_3$ Reconf} \\ \cref{lem:E3SAT-BCSP}};
    \node[problem, minimum width=3cm] at (4,0) (BCSP1){\prb{BCSP$_6$ Reconf} \\ {\footnotesize bounded degree} \\ \cref{lem:BCSP-deg_reduce}};
    \node[problem, minimum width=3cm] at (8,2.25) (E3SAT1){\prb{E$3$-SAT$(B)$ Reconf} \\ {\footnotesize bounded occurrence} \\ \cref{thm:E3SAT}};
    \node[problem, minimum width=3cm] at (12,2.25) (NCL){\prb{Nondeterministic} \\ \prb{Constraint Logic} \\ \cref{thm:NCL}};
    \node[problem, minimum width=3cm] at (12,0) (ISR){\prb{Independent Set} \\ \prb{Reconf} \\ \cref{cor:NCL-IS}};
    \node[problem, minimum width=2cm] at (15.5,0) (CR){\prb{Clique} \\ \prb{Reconf} \\ \cref{cor:IS-C}};
    \node[problem, minimum width=2cm] at (15.5,2.25) (VCR){\prb{Vertex Cover} \\ \prb{Reconf} \\ \cref{cor:IS-VC}};
    \node[problem, minimum width=3cm] at (8,0) (2SAT1){\prb{$2$-SAT$(B)$ Reconf} \\ {\footnotesize bounded occurrence} \\ \cref{cor:E3SAT-2SAT}};

    \node[] at (4,3.4) {degree reduction};
    \node[] at (BCSP.north west) (LT){};
    \node[] at (BCSP.north east) (RT){};
    \node[] at (BCSP1.south west) (LB){};
    \node[] at (BCSP1.south east) (RB){};
    \draw[dotted, ultra thick, draw=blue!65!black, fill=none, opacity=1] \convexpath{LB,LT,RT,RB}{0.35cm};
    
    \node[] at (13.5,3.6) {applications};
    \node[] at (2SAT1.south west) (LB){};
    \node[] at (2SAT1.north west) (LM){};
    \node[] at (ISR.north -| NCL.west) (CM){};
    \node[] at (NCL.north west) (CT){};
    \node[] at (VCR.north east) (RT){};
    \node[] at (CR.south east) (RB){};
    \draw[draw=none, fill=blue!65!black, opacity=0.25] \convexpath{LB,LM,CM,CT,RT,RB}{0.35cm};
    
    \node[problem, minimum width=3cm] at (12,2.25) (NCL){\prb{Nondeterministic} \\ \prb{Constraint Logic} \\ \cref{thm:NCL}};
    \node[problem, minimum width=3cm] at (12,0) (ISR){\prb{Independent Set} \\ \prb{Reconf} \\ \cref{cor:NCL-IS}};
    \node[problem, minimum width=2cm] at (15.5,0) (CR){\prb{Clique} \\ \prb{Reconf} \\ \cref{cor:IS-C}};
    \node[problem, minimum width=2cm] at (15.5,2.25) (VCR){\prb{Vertex Cover} \\ \prb{Reconf} \\ \cref{cor:IS-VC}};
    \node[problem, minimum width=3cm] at (8,0) (2SAT1){\prb{$2$-SAT$(B)$ Reconf} \\ {\footnotesize bounded occurrence} \\ \cref{cor:E3SAT-2SAT}};

    \foreach \v / \w in {qCSP/E3SAT, E3SAT/BCSP, BCSP/BCSP1, BCSP1/E3SAT1, E3SAT1/NCL, NCL/ISR, ISR/VCR, ISR/CR}
        \draw[edge] (\v)--(\w);
    \draw[edge,dotted] (E3SAT1)--(2SAT1);

\end{tikzpicture}

%% file: example1.tex
\begin{tikzpicture}
    \tikzset{edge/.style={black, thick}};
    \tikzset{node/.style={circle, thick, draw=black, fill=white, font=\LARGE, text centered, inner sep=0, outer sep=0, minimum size=0.8cm}};

    \node[node] at (0,0) (w){$\aaa$};
    \node[node] at (2,0) (v){$\aaa$};
    \node[node] at (4,0) (x){$\aaa$};
    \node[node] at (6,0) (y){$\aaa$};
    \node[node] at (0.75,1.25) (z1){\normalsize free};
    \node[node] at (1.75,1.25) (z2){\normalsize free};
    \node[text centered] at (2.5,1.25) { $\cdots$};
    \node[node] at (3.25,1.25) (zn){\normalsize free};
    
    \node[below=0.05cm of w] {\LARGE $w$};
    \node[below=0.05cm of v] {\LARGE $v$};
    \node[below=0.05cm of x] {\LARGE $x$};
    \node[below=0.05cm of y] {\LARGE $y$};
    \node[above=0.05cm of z1] {\LARGE $z_1$};
    \node[above=0.05cm of z2] {\LARGE $z_2$};
    \node[above=0.05cm of zn] {\LARGE $z_n$};

    \node[align=center,anchor=south] at (1,0) {$\pi_{(w,v)}$};
    \node[align=center,anchor=south] at (3,0) {$\pi_{(v,x)}$};
    \node[align=center,anchor=south] at (5,0) {$\pi_{(x,y)}$};
    
    \node[align=center,anchor=north] at (1,0) {$\begin{Bmatrix}(\aaa,\aaa)\end{Bmatrix}$};
    \node[align=center,anchor=north] at (3,0) {$\begin{Bmatrix}(\aaa,\aaa)\\(\bbb,\aaa)\\(\bbb,\bbb)\\(\bbb,\ccc)\\(\aaa,\ccc)\end{Bmatrix}$};
    \node[align=center,anchor=north] at (5,0) {$\begin{Bmatrix}(\aaa,\aaa)\\(\bbb,\aaa)\\(\bbb,\bbb)\\(\ccc,\bbb)\\(\ccc,\ccc)\end{Bmatrix}$};

    \foreach \v / \w in {w/v, v/x, x/y, v/z1, v/z2, v/zn}
        \draw[edge] (\v)--(\w);
\end{tikzpicture}

%% file: example2.tex
\begin{tikzpicture}
    \tikzset{edge/.style={black, thick}};
    \tikzset{node/.style={circle, thick, draw=black, fill=white, font=\LARGE, text centered, inner sep=0, outer sep=0, minimum size=0.8cm}};

    \node[node] at (0,0) (w){$\aaa$};
    \node[node] at (6,0) (x){$\aaa$};
    \node[node] at (8,0) (y){$\aaa$};
    \node[node] at (2,0) (vw){$\aaa$};
    \node[node] at (4,0) (vx){$\bbb$};
    \node[node] at (1.75,1.5) (vz1){$\aaa\bbb$};
    \node[node] at (2.75,1.5) (vz2){$\aaa\bbb$};
    \node[text centered] at (3.5,1.5) {$\cdots$};
    \node[node] at (4.25,1.5) (vzn){$\aaa\bbb$};
    \node[node] at (1.75,3) (z1){\normalsize free};
    \node[node] at (2.75,3) (z2){\normalsize free};
    \node[text centered] at (3.5,3) {$\cdots$};
    \node[node] at (4.25,3) (zn){\normalsize free};
    
    \draw[draw=none, fill=cyan!85!black, opacity=0.3] \convexpath{vz1,vz2,vzn,vx,vw}{0.55cm};

    \node[below=0.05cm of w] {\LARGE $w$};
    \node[below=0.05cm of vw] {\LARGE $v_w$};
    \node[below=0.05cm of vx] {\LARGE $v_x$};
    \node[below=0.05cm of x] {\LARGE $x$};
    \node[below=0.05cm of y] {\LARGE $y$};
    \node[text centered] at (1.3,2.2) {\LARGE $v_{z_1}$};
    \node[text centered] at (2.3,2.2) {\LARGE $v_{z_2}$};
    \node[text centered] at (3.8,2.2) {\LARGE $v_{z_n}$};
    \node[above=0.05cm of z1] {\LARGE $z_1$};
    \node[above=0.05cm of z2] {\LARGE $z_2$};
    \node[above=0.05cm of zn] {\LARGE $z_n$};

    \node[node] at (2,0) (vw){$\aaa$};
    \node[node] at (4,0) (vx){$\bbb$};
    \node[node] at (1.75,1.5) (vz1){$\aaa\bbb$};
    \node[node] at (2.75,1.5) (vz2){$\aaa\bbb$};
    \node[text centered] at (3.5,1.5) {$\cdots$};
    \node[node] at (4.25,1.5) (vzn){$\aaa\bbb$};

    \node[text centered] at (3,0) {\normalsize no edge};
    \node[align=left, anchor=center] at (5.4,0.75) {\LARGE $\cloud(v)$};
    
    \foreach \v / \w in {w/vw, vx/x, x/y, vw/vz1, vw/vz2, vw/vzn, vx/vz1, vx/vz2, vx/vzn, vz1/vz2, z1/vz1, z2/vz2, zn/vzn}
        \draw[edge] (\v)--(\w);
\end{tikzpicture}

%% file: NCL.tex
\begin{tikzpicture}
    [decoration={markings, mark=at position 0.6 with {\arrow{Straight Barb[round, scale=0.75]}}}]


    \tikzset{node/.style={circle, thick, draw=black, fill=white, font=\LARGE, text centered, inner sep=0, outer sep=0, minimum size=0.4cm}};
    \tikzset{terminator/.style={diamond, thick, draw=black, fill=white, text centered, minimum size=1cm}};
    \tikzset{edge/.style={line width=1.25mm, postaction={decorate}}};
    \tikzset{red-edge/.style={red!80!black, line width=1.25mm, postaction={decorate}}};
    \tikzset{blue-edge/.style={blue!80!black, line width=1.75mm, postaction={decorate}}};

    \node[node, minimum size=1cm] at (0,0.5) (w){$v_w$};
    \node[node, minimum size=1cm] at (4,0.5) (x){$v_x$};
    \node[node, minimum size=1cm] at (8,0.5) (y){$v_y$};
    \node[node, minimum size=1cm] at (12,0.5) (z){$v_z$};
    
    \node[node, minimum size=1cm] at (1.5,11) (C1){$v_{C_1}$};
    \node[node, minimum size=1cm] at (6,11) (C2){$v_{C_2}$};
    \node[node, minimum size=1cm] at (10.5,11) (C3){$v_{C_3}$};

    \node[terminator, below=1cm of w](w-end){free};
    \node[terminator, below=1cm of x](x-end){free};
    \node[terminator, below=1cm of y](y-end){free};
    \node[terminator, below=1cm of z](z-end){free};

    \node[terminator] at (2,1.5) (wn-end){free};
    \node[node] at (0,2.5) (wp-RB){};
    \node[node] at (0.1,5) (wp-FO){};
    \node[node] at (0.5,8)(wp1){};
    \node[node] at (4,9) (wp2){};
        
    \node[node] at (3.6,3) (xp-RB){};
    \node[node] at (3.2,5.7) (xp-FO){};
    \node[node] at (2,8.2) (xp1){};
    \node[node] at (8.5,9) (xp3){};
    \node[node] at (5.5,8.7) (xn2){};

    \node[node] at (4,7.5) (yp1){};
    \node[node] at (10,8) (yn3){};

    \node[terminator] at (14,1.5) (zn-end){free};
    \node[node] at (12,2.7) (zp-RB){};
    \node[node] at (11.8,4.8) (zp-FO){};
    \node[node] at (8.6,7.7) (zp2){};
    \node[node] at (11.5,7.6) (zp3){};

    \draw[draw=none, fill=black!25!white] \convexpath{x,xp-RB,xp-FO,xp1,C1,xp1,xp-FO,xp3,C3,xp3,xp-FO,xp-RB}{0.4cm};

    \node[node, minimum size=1cm] at (4,0.5) (x){$v_x$};
    \node[node, minimum size=1cm] at (1.5,11) (C1){$v_{C_1}$};
    \node[node, minimum size=1cm] at (10.5,11) (C3){$v_{C_3}$};
    \node[node] at (3.6,3) (xp-RB){};
    \node[node] at (3.2,5.7) (xp-FO){};
    \node[node] at (2,8.2) (xp1){};
    \node[node] at (8.5,9) (xp3){};

    \node at (-0.5,1.3) {\LARGE $w$};
    \node at (1.3,0.6) {\LARGE $\bar{w}$};
    \node at (3.3,1.3) {\LARGE $x$};
    \node at (4.5,1.3) {\LARGE $\bar{x}$};
    \node at (7,1.3) {\LARGE $y$};
    \node at (8.6,1.3) {\LARGE $\bar{y}$};
    \node at (11.5,1.3) {\LARGE $z$};
    \node at (13,0.5) {\LARGE $\bar{z}$};

    \node[above=0.25cm of C1] {\Large $C_1 = (w \vee x \vee y)$};
    \node[above=0.25cm of C2] {\Large $C_2 = (w \vee \bar{x} \vee z)$};
    \node[above=0.25cm of C3] {\Large $C_3 = (x \vee \bar{y} \vee z)$};

    \foreach \v / \w in {w-end/w, wp-RB/w, w/wn-end, wp1/wp-FO, wp2/wp-FO,    x-end/x, x/xp-RB, xp-FO/xp1, xp-FO/xp3, xn2/x}
        \draw[red-edge] (\v) -- (\w);
    \foreach \v / \w in {y-end/y, y/yp1, yn3/y,    z-end/z, zn-end/z, z/zp-RB, zp-FO/zp2, zp-FO/zp3}
        \draw[red-edge] (\v) -- (\w);
    
    \foreach \v / \w in {wp-FO/wp-RB, xp-RB/xp-FO, zp-RB/zp-FO}
        \draw[blue-edge] (\v) -- (\w);
    \foreach \v / \w in {C1/wp1, xp1/C1, yp1/C1, C2/wp2, C2/xn2, zp2/C2, xp3/C3, C3/yn3, zp3/C3}
        \draw[blue-edge] (\v) -- (\w);

\end{tikzpicture}

%% file: AND.tex
\begin{tikzpicture}
    \tikzset{edge/.style={black, thick}};
    \tikzset{node/.style={circle, thick, draw=black, fill=white, font=\LARGE, text centered, inner sep=0, outer sep=0, minimum size=1cm}};
    \tikzset{token/.style={circle, thick, fill=black, font=\LARGE, text centered, inner sep=0, outer sep=0, minimum size=0.7cm}};

    \def\sq{1.7320508}
    \node[node] at (-1,0) (x1){};
    \node[node] at (-1-\sq,-1) (x2){};
    \node[node] at (+1,0) (y1){};
    \node[node] at (+1+\sq,-1) (y2){};
    \node[node] at (0,\sq) (z1){};
    \node[node] at (0,2+\sq) (z2){};
    
    \node[token] at (-1,0) {};
    \node[token] at (+1,0) {};
    \node[token] at (0,2+\sq) {};
    
    \node[text centered] at (0,4.7) {\LARGE token edge};
    
    \foreach \v / \w in {y1/z1, z1/x1, x1/x2, y1/y2, z1/z2}
        \draw[edge] (\v)--(\w);
    
    \draw[dashed, draw=black, thick] \convexpath{x1,x2}{0.7cm};
    \draw[dashed, draw=black, thick] \convexpath{y1,y2}{0.7cm};
    \draw[dashed, draw=black, thick] \convexpath{z1,z2}{0.7cm};
    \draw[dotted, draw=gray, line width=1.2mm] \convexpath{x1,z1,y1}{1cm};
\end{tikzpicture}

%% file: OR.tex
\begin{tikzpicture}
    \tikzset{edge/.style={black, thick}};
    \tikzset{node/.style={circle, thick, draw=black, fill=white, font=\LARGE, text centered, inner sep=0, outer sep=0, minimum size=1cm}};
    \tikzset{token/.style={circle, thick, fill=black, font=\LARGE, text centered, inner sep=0, outer sep=0, minimum size=0.7cm}};

    \def\sq{1.7320508}
    \node[node] at (-1,0) (x1){};
    \node[node] at (-1-\sq,-1) (x2){};
    \node[node] at (-1-2*\sq,-2) (x3){};
    \node[node] at (+1,0) (y1){};
    \node[node] at (+1+\sq,-1) (y2){};
    \node[node] at (+1+2*\sq,-2) (y3){};
    \node[node] at (0,\sq) (z1){};
    \node[node] at (0,2+\sq) (z2){};
    \node[node] at (0,4+\sq) (z3){};
    
    \node[token] at (-1-\sq,-1) {};
    \node[token] at (+1+\sq,-1) {};
    \node[token] at (0,\sq) {};
    \node[token] at (0,4+\sq) {};
    
    \node[text centered] at (0,-1.2) {\LARGE token triangle};
    \node[text centered] at (0,6.8) {\LARGE token edge};
    
    \foreach \v / \w in {x1/y1, y1/z1, z1/x1, x1/x2, x2/x3, y1/y2, y2/y3, z1/z2, z2/z3}
        \draw[edge] (\v)--(\w);
    
    \draw[dashed, draw=black, thick] \convexpath{x2,x3}{0.7cm};
    \draw[dashed, draw=black, thick] \convexpath{y2,y3}{0.7cm};
    \draw[dashed, draw=black, thick] \convexpath{z2,z3}{0.7cm};
    \draw[dashed, draw=black, thick] \convexpath{x1,z1,y1}{0.7cm};
    \draw[dotted, draw=gray, line width=1.2mm] \convexpath{x2,z2,y2}{1cm};
\end{tikzpicture}